\documentclass[10pt]{IEEEtran}
\onecolumn
\pdfoutput=1
\def\BibTeX{{\rm B\kern-.05em{\sc i\kern-.025em b}\kern-.08em
    T\kern-.1667em\lower.7ex\hbox{E}\kern-.125emX}}
\usepackage{graphicx}
\usepackage{color}
\usepackage[]{algorithm}
\usepackage{algorithmic}
\newtheorem{theorem}{Theorem}

\newtheorem{lemma}[theorem]{Lemma}
\newtheorem{corollary}[theorem]{Corollary}

\newtheorem{definition}{Definition}
\newtheorem{note}{Note}

\newcommand{\1}{{\bf 1}} 
\newcommand{\E}{\mathsf{E}} 

\newcommand{\card}[1]           {\left| #1\right|}

\newcommand\argmin{\mathop{\mbox{{\rm argmin}}}\limits}

\newcommand{\bea}{\begin{eqnarray}}
\newcommand{\eea}{\end{eqnarray}}
\newcommand{\beas}{\begin{eqnarray*}}
\newcommand{\eeas}{\end{eqnarray*}}

\begin{document}
\title{On Real Time Coding with Limited Lookahead}
\author{\IEEEauthorblockN{Himanshu Asnani\IEEEauthorrefmark{1} and 
Tsachy Weissman\IEEEauthorrefmark{2}} 
\thanks{\IEEEauthorblockA{\IEEEauthorrefmark{1}Stanford University, Email:
asnani@stanford.edu.}}
\thanks{\IEEEauthorblockA{\IEEEauthorrefmark{2}Stanford University, Email:
tsachy@stanford.edu.}}}

\maketitle


\maketitle

\begin{abstract}
A real time coding system with lookahead consists of a 
memoryless source, a memoryless channel, an
encoder, which encodes the source symbols sequentially with knowledge
of future source symbols upto a fixed finite lookahead, $d$, with or without
feedback of the past channel output symbols and a decoder, which sequentially
constructs the source symbols using the channel output. The objective is to
minimize the expected
per-symbol distortion.
\par
For a fixed finite lookahead $d\ge 1$ we invoke the
theory of controlled markov chains to obtain an average cost optimality
equation (ACOE), the solution of which, denoted by $D(d)$, is the minimum
expected per-symbol
distortion. With increasing $d$, $D(d)$ bridges the gap between causal encoding,
$d=0$, where symbol by
symbol encoding-decoding is optimal and the infinite lookahead case,
$d=\infty$, where Shannon Theoretic arguments show that
separation is optimal.
\par
 We extend the analysis to a system with finite state
decoders, with or without noise-free feedback. For a Bernoulli source and binary
symmetric channel, under hamming loss, we compute the optimal
distortion for
various source and channel parameters, and thus obtain computable bounds on
$D(d)$. We also identify regions of source and channel parameters where
symbol by symbol encoding-decoding is suboptimal. Finally, we demonstrate the
wide applicability of our approach by applying it in additional coding
scenarios, such as the case where the sequential decoder can take cost
constrained actions affecting the quality or availability of side information
about the source.  

\end{abstract}


\section{Introduction}
\label{intro}\par
\subsection{Motivation and Related Work}
\label{motivation}
{A} memoryless source $\{U_1, U_2, \ldots\}$  is to be communicated over a
memoryless channel with
the objective of minimizing expected average (per-symbol) distortion, with or
without the
availability of unit-delay noise-free feedback. The communication is in real
time and hence
the
encoding and decoding is sequential, with a fixed finite lookahead of
source symbols available at the encoder (cf. the setting in Fig.
\ref{feedbackmemoryrtclookahead}). The motivation stems from practical systems
such as for 
video streaming, cache memory devices in
computing systems, real time communication systems etc., where the encoder has a
fixed buffer of
future source symbols, and the quality of service demands that encoding and
decoding should be in real time. The problem finds its applications in other
sequential decision systems, where resource allocation should be done on the fly
due
to adverse effects of latency or delay, such as sensor networks,
weather-monitoring systems, flow in societal networks such as transportation
networks, recycling systems, etc. A
natural criterion of performance is to minimize the expected average distortion.
What is the best we can do here ? Note that
 such a framework with real time constraints is not covered by Shannon Theory.
In classical Information Theory, encoding of long \textquotedblleft
typical\textquotedblright\ sequences
in blocks as well as block decoding introduces large delays and thus such 
achievable schemes violate the very
premise of bounded or no delay constraint. To answer the question, we
invoke markov decision theory and cast our problem and other such variants as
discrete time controlled markov chains with average
cost criterion.
\par
The problem is well motivated by practical problems of delay constrained
source-channel coding and has been of much interest in the literature. There
have also been many different ways to model the notion of sequential encoding
and decoding. In
the source coding context, causal source codes were studied in
\cite{NeuhoffGilbert}, \cite{LinderZamir}, \cite{Piret}, which demand the
reconstruction to depend causally on the source symbols. But this is a
much weaker constraint and causal source codes can operate on large delays as
was
pointed out in \cite{NeuhoffGilbert} itself. Causal source codes with side
information were studied in \cite{WeissmanMerhavCausal}.
\par
Note that we can transform our setting of limited encoder lookahead of $d$, to
that of a zero lookahead of a markov source,
$V_i=U_{i}^{i+d}$. This transformation puts the problem in the class of
sequential encoding decoding problems with markov sources. 
When the communication horizon is fixed, the structure of optimal encoding and
decoding policies with Markov sources have been studied in \cite{Witsenhausen},
\cite{TeneketzisPhD}, \cite{TeneketzisRTC}, \cite{WalrandVaraiya},
\cite{MunsonPHD},  \cite{GorantlaColeman}.  In \cite{MahajanTeneketzis}, authors
propose a systematic methodology 
for such a non-classical information structure to search for an optimal strategy. 
\par
The problem of real time coding and decoding in semi stochastic setting, 
i.e., for the individual sequences was studied in \cite{LinderLugosi} and
\cite{TsachyMerhavlimitedelay}, while finite state digital systems were the
subject of study in \cite{GaarderSlepian}.
\par
The connection between dynamic programming and information theory 
has been well exploited. The problem of computing the capacity of channels with
feedback was formulated as a Markov Decision Process in \cite{TatikondaPHD},
\cite{TatikondaMitter}. The long standing problem of capacity of trapdoor
channel
(cf. \cite{Blackwell}, \cite{Ash}) with feedback was evaluated using average
cost optimality equations in \cite{HaimCuffRoyTsachy}.  Zero error capacity for
certain channel coding problems was computed using dynamic programming in
\cite{ZhaoHaimZeroError}.

\subsection{Contributions and Organization of the Paper}
\label{contributionorganization}
The approaches in \cite{Witsenhausen},
\cite{TeneketzisPhD}, \cite{TeneketzisRTC}, \cite{WalrandVaraiya},
\cite{MunsonPHD},  \cite{GorantlaColeman} and \cite{MahajanTeneketzis} are 
inspired by control theory, which provides tools for finding optimal schemes
and understanding their structure. In this work, we take these tools further to
provide more explicit expressions and bounds for the optimum performance under a
given lookahead constraint $d$. While optimum performance in the case $d=0$ is
easily shown to be attained by \textquotedblleft
symbol-by-symbol\textquotedblright\ operations, and the case $d=\infty$ can be
answered with the tools of Shannon theory, for any finite $d\ge 1$, the existing
literature does not provide useful analytical values or bounds on the
minimum expected average distortion, $D(d)$. In addition to being amenable to
a decision theoretic formulation of Markov sources, as in the surveyed
literature
above, the model we consider here is more basic and lends itself to simpler
average 
cost
optimality equation, which in some cases (cf. Section
\ref{rtcfinite}) can be computed exactly. While in \cite{TeneketzisRTC},
\cite{WalrandVaraiya}, \cite{GorantlaColeman} emphasis is on expected total
fixed horizon cost, we argue that expected average cost over infinite horizon is
a more natural criterion of
performance as in the sequential encoding and decoding problems, we
typically do not know when to stop, and hence we would like to analyze the
asymptotics of the horizon-independent problem. While the main focus in this
work has been to
characterize the minimum achievable distortion, the average cost
optimality
equations also characterize sufficient conditions on the optimality of
stationary (encoding and decoding) policies. 
\par
Note that in our
communication problem in Fig.
\ref{feedbackmemoryrtclookahead}, the lookahead is available only at the
encoder while the decoder constructs the estimates causally, instead of a
seemingly more general setting where lookahead of $l_e$ is present at the
encoder while decoder has lookahead $l_d$. However performance of any
policy/code with encoder and decoder lookahead parameters $(l_e,l_m)$ can be
attained arbitrarily closely by the optimal policy for our setting in Fig.
\ref{feedbackmemoryrtclookahead} with $d=l_e+l_m$ as pointed out in Section
II of \cite{TsachyMerhavlimitedelay}.
Authors in
\cite{SahaiBlock} consider the communication problem similar to our setting with
$l_e=0$, $l_d=d$ for $d\ge 0$, per-symbol distortion $D(d)$ and show that
$D(d)$ converges exponentially rapidly to $D(\infty)$ and provide bounds on the
exponent. However the results are asymptotic in nature and hence different from
this work, which is explicit exact or approximate characterization of values for
$D(d)$ for any fixed, possibly small $d$. 
\par
Recently there has been work in the
direction of \textquotedblleft action in information theory\textquotedblright\
, i.e. canonical Shannon theoretic models with encoder and/or decoder
taking cost constrained actions to affect the generation or availability of
channel state information, side
information,
feedback etc., cf.
action in point to point scenarios in \cite{HaimTsachyVendor},
\cite{TsachyChannel}, \cite{Kittichokechai} \cite{HimanshuHaimTsachyProbing},
\cite{HimanshuHaimTsachyFeedback} and in
multi-terminal systems in \cite{ChiaHimanshuTsachy},
\cite{HaimHimanshuCribbing}. We revisit the setting of source coding with a side
information vending machine, as in \cite{HaimTsachyVendor} (See Fig.
\ref{actionfeedbackmemoryrtclookahead}) for the case where the
encoding is sequential with lookahead, decoder takes an action $A_v$
sequentially
dependent on the encoded symbols to get side information about the source
through a memoryless channel, $P_{Y|U,A_v}$. The reconstruction of the source
is
based upon the current encoded symbol, the current side information symbol and
memories storing the past encoded symbols and side information symbols. We show
that the problem can be formulated as a constrained Markov Decision
Process.
\par 
The main contribution of this paper is the casting of a large class of limited
delay source, channel and joint source-channel coding problems in the realm of
sequential decision theory, obtain characterizations of the optimum
performance via average
cost optimality equations with finite or compact state spaces, and solve exactly
or obtain bounds for the expected average distortion as a function of
lookahead $d$.
\par
The paper is organized as follows. 
Section \ref{problem} describes the basic model of problems with lookahead (See
Fig. \ref{feedbackmemoryrtclookahead}), encoding is sequential using the
lookahead and unit delay noise-free feedback, $X_i(U^{i+d},Y^{i-1})$, while the
decoding depends on the current channel output and the past memory,
$\hat{U}_i(Y_i,Z_{i-1})$. The memory evolves as
$Z_i(Z_{i-1},Y_i)$. We seek to find the minimum expected average distortion as a
function of lookahead, i.e.,
\bea
D(d)=\inf_{\{X_i(\cdot)\},\{\hat{U}_i(\cdot)\}}\limsup_
{N\rightarrow\infty } \E\left [ \frac { 1 } { N } \sum_
{i=1}^N\Lambda(U_i,
\hat{U}_i)\right].
\eea
In Section \ref{controlledmarkov} we present an overview of controlled 
markov processes with average cost, the unconstrained case in
Section \ref{unconstrainedcontrol} and constrained control in Section
\ref{constrainedcontrol}. Section \ref{rtccomplete} studies the case of complete
memory,
i.e.,
$Z_{i}=Y^i$. In Section \ref{rtccompleteacoe} we use the theory of Section
\ref{controlledmarkov} to construct an average cost optimality equation, the
solution to which is the average optimal distortion. In Section
\ref{rtccompletelookahead}, we consider the question \textquotedblleft to
look or not to lookahead \textquotedblright\ and specify a sufficient condition
under which symbol by symbol encoding-decoding is optimal for a given source,
channel, distortion function and lookahead. This kind of result in our
problem of sequential encoding decoding with lookahead complements that of
\textquotedblleft to code or not to code \textquotedblright\ of
\cite{GastparRimoldiVetterli}. In Section \ref{rtcfinite}, we consider the
framework with finite state decoders, constructing corresponding ACOE in
Section \ref{rtcfiniteacoe}. In Section \ref{rtcfinitecomputation}, we
use relative value iteration to solve the problem exactly for an example of
binary source and binary symmetric channel under hamming loss, thereby
demonstrating 
how the average
distortion values for this setting can be used to bound $D(d)$ of Section
\ref{rtccomplete}. We also contrast with the extreme cases of no lookahead,
$d=0$, where symbol by symbol policies are optimal and $d=\infty$ where
Shannon's 
Separation Theorem \cite{Shannon} determines the minimum expected average
distortion. We also highlight the regions of source-channel parameters where for
any finite $d\ge 1$, symbol by symbol encoding-decoding is strictly suboptimal
for a Bernoulli
source and binary symmetric channel. Section \ref{nofeedbackrtc}
relaxes the assumption of the previous sections that feedback is present. In
Section
\ref{rtcaction}, the setting of source coding with a side information vending
machine is
considered. Here again, encoding is sequential with lookahead, decoder takes
cost constrained actions, $A_{v,i}$, sequentially to get side
information about the source through a memoryless channel, $P_{Y|U,A_v}$. The
decoding is the
optimal reconstruction $\hat{U}_i(X_i,Y_i,M_{i-1},N_{i-1})$, where $M_{i-1}$ and
$N_{i-1}$ are the memories storing some or all of past encoded symbols and side
information symbols, respectively. Section \ref{sirtcaction} evaluates
the case when encoder also has access to the side information, with
decoder having complete
memory in Section \ref{sirtcactioncomplete}, while finite memory decoders are
considered in Section 
\ref{sirtcactionfinite}. Section \ref{nosirtcaction} studies the same source
coding problem with a side information vending machine but now encoder has no
access to side information. Section
\ref{summary} summarizes the methodology developed in this paper of constructing
average cost
optimality equations. The paper
is concluded in
Section \ref{conclusion}.

\section{Problem Formulation}
\label{problem}
We begin by explaining the notation to be used throughout this paper.
Let upper case, lower case, and calligraphic letters denote, respectively,
random
variables, specific or deterministic values which random variables may assume,
and
their alphabets. For two jointly distributed random variables, $X$ and $Y$, let
$P_X$, $P_{XY}$ and $P_{X|Y}$ respectively denote the marginal of $X$, joint
distribution of $(X,Y)$ and conditional distribution of 
$X$ given $Y$. $X_{m}^{n}$ is a shorthand for the $n-m+1$ tuple
$\{X_m,X_{m+1},\cdots,X_{n-1},X_n\}$.
$\mathcal{B}(\mathcal{X})$ denotes the Borel
$\sigma$-algebra of a given topological space,
$\mathcal{X}$. $\mathcal{P}(\mathcal{X})$ denotes the
probability simplex on the finite alphabet, $\mathcal{X}$. $C_b(\mathcal{X})$
denotes the set of continuous and bounded functions on the topological
space $\mathcal{X}$. $\1_{\{\cdot\}}$ stands for the indicator function.
$\mathbf{N}$ and $\mathbf{R}$ denote the sets of natural and real numbers
respectively. We impose the
assumption of finiteness of
cardinality on all alphabets of operational significance (source, channel input,
channel output, reconstruction), unless otherwise indicated.
\begin{figure}[htbp]
\begin{center}
\scalebox{0.7}{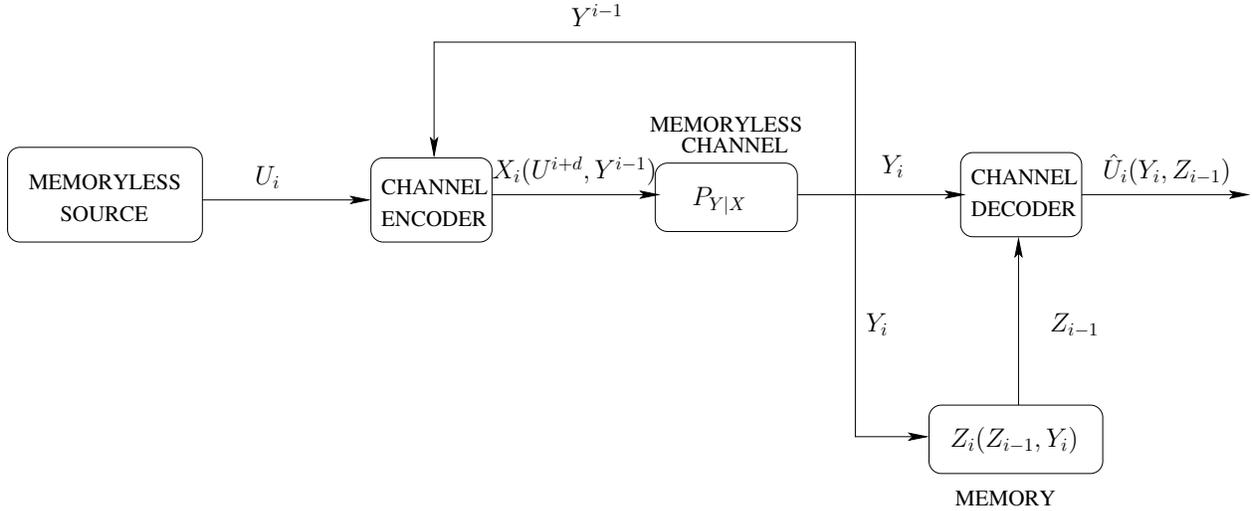}
\caption{Real time coding with lookahead. Encoder uses future source
symbols upto a fixed finite lookahead, $d$ and unit-delay noise free feedback,
decoder uses present channel output and past memory for source reconstruction.
Complete memory case corresponds to 
$Z_i=Y^i$ which implies
$\card{\mathcal{Z}_i}=\card{\mathcal{Y}}^i$.}
\label{feedbackmemoryrtclookahead}
\end{center}
\end{figure}
The general problem setup, depicted in Fig.
\ref{feedbackmemoryrtclookahead} consists
of the following principle components :
\begin{itemize}
 \item \textit{Source} : Generates i.i.d. source symbols,
$\{U_i\}_{i\in\mathbf{N}} \in \mathcal{U}$. The source symbols are
distributed $\sim P_U$. 
\item \textit{Channel Encoder} : The encoder has access to unit-delay
noise-free feedback
from the channel output and future source symbols upto a fixed finite lookahead,
$d$,
i.e,
$X_i=f_{e,i}(U^{i+d},Y^{i-1})$, where $f_{e,i}$ is the encoding
function, $f_{e,i}:\mathcal{U}^{i+d}\times\mathcal{Y}^{i-1}\rightarrow
\mathcal{X}$, $i\in\mathbf{N}$. 
\item \textit{Channel} : Given channel input symbol, $x_i$, and all the source
symbols and past
channel inputs and outputs, $(u_1^{\infty},x^{i-1},y^{i-1})$, channel output,
$y_i$
is
distributed i.i.d. $\sim P_{Y|X}$, i.e.,
\bea\label{DMC}
P(y_i|u_1^{\infty},x^i,y^{i-1})=P_{Y|X}(y_i|x_i).
\eea 

\item \textit{Memory} : The decoder cannot make use of all the channel output
symbols upto current time due to memory constraints. Memory is updated as a
function of the past state of the memory and the current channel output, i.e.,
$Z_i=f_{m,i}(Z_{i-1},Y_i)$, where the $f_{m.i}$ is the memory update
function,
$f_{m,i}:\mathcal{Z}_{i-1}\times\mathcal{Y}\rightarrow\mathcal{Z}_{i}$,
$i\in\mathbf{N}$. Note that the alphabet $\mathcal{Z}_i$ can grow with $i$,
hence the setup also includes the special case of complete memory, i.e.,
$Z_i=Y^i$ which implies
$\card{\mathcal{Z}_i}=\card{\mathcal{Y}}^i$.
\item \textit{Channel Decoder} : Channel decoder uses the current channel
output and the past memory state to construct its estimate of the source
symbol, i.e., $\hat{U}_i=f_{d,i}(Z_{i-1},Y_i)$, the decoding rule is
the map, $f_{d,i} :
\mathcal{Z}_{i-1}\times\mathcal{Y}\rightarrow\hat{\mathcal{U}}$.
\end{itemize}
The alphabets $\mathcal{U}$, $\mathcal{X}$, $\mathcal{Y}$ and $\hat\mathcal{U}$
are assumed to be finite. Let
$\Lambda(\cdot,\cdot):\mathcal{U}\times\hat\mathcal{U}\rightarrow\mathbf{R}$
 indicate a distortion function. We assume for simplicity that, $0\le
\Lambda(\cdot,\cdot)\le\Lambda_{max}<\infty$. Let the tuple
$\mu(d)=(f_e,f_m,f_d)$ indicate the sequence of encoding rules,
$\{f_{e,i}\}_{i\in\mathbf{N}}$, memory update rules,
$\{f_{m,i}\}_{i\in\mathbf{N}}$ and decoding rules,
$\{f_{d,i}\}_{i\in\mathbf{N}}$.
\begin{definition}\label{def1}[Distortion-Optimal Policy]
For a fixed lookahead, $d$, we define \textit{$d$-distortion optimal} policies,
$\mathcal{P}^{opt}(d)$ as the set of $(f_e,f_m,f_d)$-policies, denoted by
$\mu(d)$, which achieve the minimum expected
average distortion, i.e,
\bea
\mathcal{P}^{opt}(d)=\left\lbrace\mu(d):\mu(d)=\arg\inf_{\{f_e,f_m,f_d\}}
\limsup_
{N\rightarrow\infty } \E\left [ \frac { 1 } { N } \sum_
{i=1}^N\Lambda(U_i,
\hat{U}_i)\right]\right\rbrace.
\eea
The corresponding minimum expected distortion  as a function of lookahead, $d$,
\bea\label{dopt}
D(d)=\inf_{\{f_e,f_m,f_d\}}\limsup_{N\rightarrow\infty}\E\left[\frac{1}{N}
\sum_
{i=1}^N\Lambda(U_i,
\hat{U}_i)\right].
\eea
\end{definition}
Our main goal is to characterize $D(d)$ and identify structural properties of
the elements of 
$\mathcal{P}^{opt}(d)$.

\begin{note}\label{note1}
Note that $\inf$ in the definition of $D(d)$ can equivalently be replaced by
$\min$ (cf. Appendix \ref{appendix0}). This implies that $\mathcal{P}^{opt}(d)$
is non-empty. Taking limsup in
definition of $D(d)$, while appearing more conservative, is actually
inconsequential as you would get the same value of D(d) if you put a liminf in
the definition. This can be easily argued as follows. Let, the
per-symbol expected distortion under a policy $\mu$ upto time $N$ be denoted by
$D_{\mu}^{(N)}$. Denoting $D^{\sup}(d)$ and $D^{\inf}(d)$ as the distortion
criterion with $\limsup$ and $\liminf$ respectively, we know $D^{\inf}(d)\le
D^{\sup}(d)$. We will now show $D^{\inf}(d)\ge
D^{\sup}(d)$. Let a policy $\mu^{\ast}$ attains the infimum for
$D^{\inf}(d)$
(that there exists such policy follows from the same arguments as above for the
non-emptiness of $\mathcal{P}^{opt}(d)$). This implies (as $\Lambda(\cdot)$ is
bounded) for $\epsilon>0$, $\exists$ $N(\epsilon)>0$ such that under this policy
$D^{(N(\epsilon))}\le D^{\inf}(d)+\epsilon$. Operating such a policy in $b$
blocks, 
\bea
D^{\sup}(d)\le\lim_{b\rightarrow\infty}D_{\mu^{\ast}}^{(N(\epsilon)b)}\le
D^{\inf}(d)+\epsilon,
\eea
which implies in the limit $\epsilon \rightarrow 0$, $D^{\sup}(d)\le
D^{\inf}(d)$.
\end{note}

\section{Controlled Markov Process with Average Cost : Background and Preliminaries}
\label{controlledmarkov}
We present here an overview of parts of the controlled Markov process with
average cost
criterion framework that will be applied. First, we present an overview of the
unconstrained case where
the only objective is to maximize an expected average cost. We then
consider the
constrained case where, in addition, the system needs to satisfy certain
expected
average cost constraints. 
\subsection{Unconstrained Control}
\label{unconstrainedcontrol}
Here we overview results about general Borel state and action spaces. We refer
to 
\cite{Borkar_survey} for a more complete discussion. The
problem is characterized by the tuple
$(\mathcal{S},A_s, \mathcal{A}, \mathcal{W},F,P_{S},P_{W},g)$ and a discrete
time
dynamical system,
\bea
s_t=F(s_{t-1},a_t,w_t),
\eea
where the states $s_t$ take values in finite, countable or in general Borel
space $\mathcal{S}$ (called the state space), actions $a_t$ take values in the
admissible action space,
$A_s(s_t)$ which is a subset of a compact subset
$\mathcal{A}$ (called the action space) of a Borel space, and the disturbance,
$w_t$, takes values in a
measurable space $\mathcal{W}$ (called the disturbance space). Initial state
$S_0$ is drawn with distribution
$P_S$ and the disturbance $w_t$ is drawn from the distribution,
$P_W(\cdot|s_{t-1},a_t)$ which depends on past actions and states, only through
the pair $(s_{t-1},a_t)$. We consider only  measurable functions. A policy
$\pi$ is defined to be the sequence of functions, $\pi=(\mu_1,\mu_2,\cdots)$,
where $\mu_t$ is the function which maps histories
($\phi_t=(s_0,w_0,\cdots,w_{t-1})$) to actions. A set of \textit{history
deterministic} policies, $\Pi_{HD}$ is characterized by policies for
which actions are generated as $a_t=\mu_t(\Phi_t)$. A set of \textit{Markov
deterministic} policies, $\Pi_{MD}$ is characterized by policies for
which actions are generated as $a_t=\mu_t(s_{t-1})$. A set of policies
$\Pi_{SD}$
is
referred to as \textit{stationary deterministic} if it is characterized by a
function $\mu:\mathcal{S}\rightarrow\mathcal{A}$ such that,
$\mu_t(\Phi_t)=\mu(s_{t-1})$ $\forall\ t$. Policies can be randomized or
deterministic (\cite{Borkar_survey}, Section 2.2). The policy sets $\Pi_{HR}$,
$\Pi_{MR}$ and $\Pi_{SR}$ respectively stand for \textit{history randomized},
\textit{markov randomized} and \textit{stationary randomized} policies. As per
our definitions and interests, the largest class of policies considered
henceforth
will be history deterministic policies, $\Pi_{HD}$. Let
\bea
\mathcal{K}=\{(x,a):x\in\mathcal{S}, a\in A_s(x)\}\in
\mathcal{B}(\mathcal{S}\times\mathcal{A}).
\eea
Note if $\mathcal{S}$ and $\mathcal{A}$ are compact subsets of a Borel space,
$\mathcal{K}$ is a compact subset
$\in\mathcal{B}(\mathcal{S}\times\mathcal{A})$. The dynamics induce a
stochastic transition kernel on
$\mathcal{B}(\mathcal{S})\times\mathcal{K}$, $Q(\cdot|x,a)$,
which implies for each $(x,a)\in\mathcal{K}$, $Q(\cdot|x,a)$
is probability measure on $\mathcal{B}(\mathcal{S})$ and for each
$D\in\mathcal{B}(\mathcal{S})$, $Q(D|\cdot)$ is Borel measurable on
$\mathcal{K}$.\\
The objective is to maximize expected average reward given a bounded one
stage reward
function,
$g:\mathcal{K}\rightarrow \mathbf{R}$ and find the optimal
policy. The average reward of a policy $\pi$ with a given initial state
distribution $\nu$ is defined by,
\bea
J({\nu,\pi})\stackrel{\triangle}{=}\liminf_{N\rightarrow\infty}\E^{\pi}_{\nu}
\left [ \frac { 1 } { N } \sum_ { t=1 }^{N}
g(S_{t-1},\mu_{t}(\Phi_{t}))\right].
\eea
The optimal average reward and the optimal policy is defined by,
\bea
J^{opt}(\nu)&=&\sup_{\pi}J({\nu,\pi})\\
\pi^{opt}(\nu)&=&\{\pi:J({\nu,\pi})=J^{opt}(\nu)\}.
\eea
Note that in general for a controlled Markov process with average cost
criterion,
where the state space is infinite, the total expected average cost might depend
on the initial state.  However, operationally, since our objective is to
minimize
the expected average distortion as in Eq. \ref{dopt}, we can decide to start of
the system
with the best initial state, state which yields the best distortion, in which
case the optimal cost and optimal policy will be denoted by, $J^{opt}$ and
$\pi^{opt}$.
\bea
J^{opt}&=&\sup_{\nu}J^{opt}({\nu})\\
\pi^{opt}&=&\{\pi:\exists\ \nu\mbox{ s.t. } J(\nu,\pi)=J^{opt}\}.
\eea
We need not dwell on sensitivity of the optimal cost to initial states, as this
will not be an issue in our application of this framework. However when state
space is say finite, irreducible and positive
recurrent, average cost is indeed equal for all initial states.
In general, there can be more than one optimal policy, in which case ties are
resolved arbitrarily.
\\
The following theorem describes
the average cost optimality equation (ACOE) for such a process, and relates the
optimal reward with the optimal stationary deterministic policy.
\begin{theorem}[cf. Theorem 6.1 of \cite{Borkar_survey}]
\label{theorem1}
If $\lambda\in\mathbf{R}$ and a bounded function
$h:\mathcal{S}\rightarrow\mathbf{R}$ satisfy,
\bea
\lambda+h(s)=\sup_{a\in\mathcal{A}}\left[g(s,a)+\int
P_W(dw|s,a)h(F(s,a,w))\right],\mbox{ }\forall\ s\in\mathcal{S},
\eea 
then $\lambda=J^{opt}$. Further, if there is a function $\mu :
\mathcal{S}\rightarrow \mathcal{A}$ such that $\mu(s)$ attains the supremum
above for all states, then $J({\pi})=J^{opt}$ for
$\pi=\{\mu_1,\mu_2,\cdots\}$ with $\mu_i(\Phi_i)=\mu(s_{i-1})$, $\forall i$. 
\end{theorem}
\begin{note}
 As in \cite{Borkar_survey}, the above theorem assumes the conditions of
semi-continuous model, (\cite{Borkar_survey}, Section 2.4). However in the set
of problems considered in our paper, all such
assumptions will be trivially met such as the transition kernel being weakly
continuous in $\mathcal{K}$ and the continuity of $g$. For brevity, we
omit
explicitly mentioning such
assumptions before invoking the above theorem in the sections to follow.
\end{note}

\subsection{Constrained Control}
\label{constrainedcontrol}
In constrained control, the
system is characterized by the tuple
$(\mathcal{S},A_s,\mathcal{A},\mathcal{W},F,P_{S},P_{W},g,\mathbf{l},\mathbf{
\Gamma } )$. With all the terms carrying the same meaning as in previous
subsection, $\mathbf{l}=\{l_1(\cdot),\cdots,l_k(\cdot)\}$ and
$\mathbf{\Gamma}=\{\Gamma_1,\cdots,\Gamma_k\}$ are respectively $k$-dimensional
constraint functions (defined on $\mathcal{K}$) and cost
vectors for some $k\in\mathbf{N}$. the dynamics of the system are
precisely the same as in the unconstrained case, the objective here
being,
\bea
\mathbf{maximize } \mbox{ }&&J(\nu,\mu)\nonumber\\
\mathbf{subject}\mbox{ }\mathbf{to}\mbox{ }&&J^{c}_i(\nu,\mu)\le\Gamma_i\mbox{
}\forall\mbox{ }i=1,\cdots,k,
\eea
where,
\bea
J({\nu,\pi})\stackrel{\triangle}{=}\liminf_{N\rightarrow\infty}\E^{\pi}_{\nu}
\left [ \frac { 1 } { N } \sum_ { t=1 }^{N}
g(S_{t-1},\mu_{t}(\Phi_{t}))\right],
\eea
is the average cost and,
\bea
J_i^c({\nu,\pi})\stackrel{\triangle}{=}\liminf_{N\rightarrow\infty}\E^{\pi}_{\nu
}
\left [ \frac { 1 } { N } \sum_ { t=1 }^{N}
l_i(S_{t-1},\mu_{t}(\Phi_{t}))\right]\mbox{ }\forall\mbox{ }i=1,\cdots,k,
\eea
are the constraints. \cite{Borkar_survey} and \cite{Altman} provide a treatment
of this problem
but only for denumerable states. We here present the more general framework of 
\cite{KuranoNakagamiHuang}, with compact state and action spaces. The
Lagrangian, $L$, associated with the
problem is defined as,
\bea
L((\nu,\pi),\lambda)=J(\nu,\pi)+\sum_{i=1}^{k}\lambda_i(\Gamma_i-J_i^c(\nu,\mu))
,
\eea
for any $(\nu,\pi)\in P(\mathcal{S})\times\Pi_{HD}$ and
$\lambda=(\lambda_1,\cdots,\lambda_k)\in\mathbf{R}^{k}_{+}$ (positive orthant
of the $k$-dimensional Euclidean space).\par
The following
theorem gives conditions of optimality of a particular
initial state distribution and a policy. 
\begin{theorem}\label{theorem2}[Theorem 2.3 of \cite{KuranoNakagamiHuang}]
Assume the following conditions for the tuple
$(\mathcal{S},A_s,\mathcal{A},\mathcal{W},F,P_{S},P_{W},g,\mathbf{l},\mathbf{
\Gamma } )$,
\begin{itemize}
 \item [C1]$\mathcal{S}$ and $\mathcal{K}$ are compact.
 \item [C2] $g\in C_b(\mathcal{K})$ and $l_i\in
C_b(\mathcal{K}), \mbox{ }\forall\mbox{ } i=1,\cdots,k$.
 \item [C3] For all $x_n\rightarrow x$ and $a_n\rightarrow a$,
$Q(\cdot|x_n,a_n)$, converges weakly to $Q(\cdot|x,a)$.
 \item [C4] (\textit{Slater's Condition})  There exists a
$(\overline{\nu},\overline{\pi})\in P(\mathcal{S})\times \Pi_{HD}$ such that,
\bea
J_i^c(\overline{\nu},\overline{\pi})<\Gamma_i \mbox{ }\forall\mbox{
}i=1,\cdots,k.
\eea
\end{itemize}
Under the conditions C1-C4, the Lagrangian $L(\cdot,\cdot)$ has a saddle point
with a randomized stationary policy, i.e., $\exists\mbox{ } \lambda^{\ast}\ge 0$
and
$(\nu^{\ast},\pi^{\ast})\in P(\mathcal{S})\times\Pi_{SR}$ such that,
\bea
L((\nu,\pi),\lambda^{\ast})\le
L((\nu^{\ast},\pi^{\ast}),\lambda^{\ast})\le
L((\nu^{\ast},\pi^{\ast}),\lambda),\mbox{ }\forall\mbox{ }(\nu,\pi)\in
P(\mathcal{S})\times\Pi, \mbox{ }\lambda\ge 0,
\eea
which implies (from Theorem 2.1 of \cite{KuranoNakagamiHuang}) that
$(\nu^{\ast},\pi^{\ast})$ is a constrained optimal pair.
Further (Theorem 2.2 of \cite{KuranoNakagamiHuang}),
\bea
L^{\ast}=L((\nu^{\ast},\mu^{\ast}),\lambda^{\ast})=\inf_{\lambda\ge
0}\sup_{(\nu,\pi)\in
P(\mathcal{S})\times\Pi_{HD}}L((\nu,\pi),\lambda)=\sup_{(\nu,\pi)\in
P(\mathcal{S})\times\Pi_{HD}} \inf_{\lambda\ge
0} L((\nu,\pi),\lambda),
\eea
and $L^{\ast}$ is the solution of the problem or the minimum expected average
distortion such that the constraints are satisfied.
\end{theorem}
\begin{note}
 \label{note2}
In all the settings considered henceforth, $\forall\ s\in\mathcal{S}$,
$A_s(s)=\mathcal{A}$, hence with benign abuse of notation, we will drop $A_s$
from the tuple associated with our description.
\end{note}

\section{Real-Time Coding with Limited Lookahead : Complete Memory}
\label{rtccomplete}
The problem we described in Section \ref{problem} (Fig.
\ref{feedbackmemoryrtclookahead}) is an abstraction of a real time
communication problem
with the encoder having a fixed lookahead of the future source symbols and a
perfect unit-delay feedback of the channel output symbols. In this section, we
show that this problem can be formulated as a controlled Markov chain process
with average cost
criterion, and derive an optimality equation. Before that, we modify our
source to concentrate on an equivalent problem. Note that the i.i.d. source,
$S=\{U_i\}_{i\in\mathbf{N}}$ considered can be replaced by a markov source
$S_{M}=\{V_i\}_{i\in\mathbf{N}}$ such that,
$V_i=U_i^{i+d}\in\mathcal{U}^{d+1}$. Since the source $S$ is i.i.d., the
transition kernel for
this Markov process $S_M$ from $v=(u_1,u_2,\cdots,u_{d+1})$ to
$\tilde{v}=(\tilde{u}_1,\tilde{u}_2,\cdots,\tilde{u}_{d+1})$ is given by,
\bea
K(v,\tilde{v})=P(\tilde{v}|v)=\1_{\{(u_2,\cdots,u_{d+1})=(\tilde{u}_1,\cdots,
\tilde{u}_d)\}}P_U(\tilde{u}_{d+1}).
\eea
The transition matrix is denoted
 by $\mathbf{K}$. Let us assume the distribution of initial state is
$P_V$. Also there is no loss of
optimality in considering encoding functions to be $\card{\mathcal{V}}$
dimensional mappings, $\{f_{e,i}(v,V^{i-1},Y^{i-1})\}_{v\in\mathcal{V}}$. The 
effective problem with modified source, $S_M$ is now a real-time communication
problem as in Fig. \ref{feedbackmemoryrtc} with no lookahead. For this modified
problem, we seek
to minimize the average reward,
\bea
\inf\limsup_{n\rightarrow\infty}\frac{1}{n}\E\left[\sum_{i=1}^{n}
\tilde{\Lambda}(V_i ,
\hat { V } ^ {
opt}_i(Y^i))\right]=\inf\limsup_{n\rightarrow\infty}\frac{1}{n}\E\left[\sum_{i=1
}^{n}
\Lambda(U_i ,
\hat {U} ^ {
opt}_i(Y^i))\right],
\eea
where $\tilde{\Lambda}(V_i ,
\hat { V } ^ {
opt}_i(Y^i))=\Lambda(U_i ,
\hat {U} ^ {
opt}_i(Y^i))$.
In this section we construct an average cost optimality equation for the
equivalent problem in Fig. \ref{feedbackmemoryrtc} and complete memory, i.e.
$Z_i=Y^{i}$. 
\begin{figure}[htbp]
\begin{center}
\scalebox{0.6}{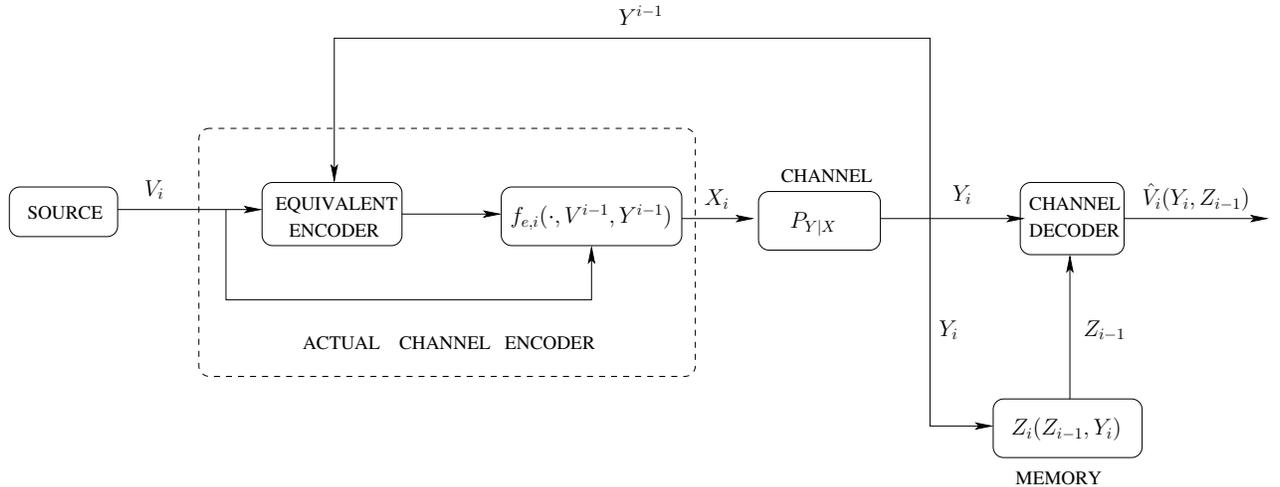}
\caption{Equivalent problem to Fig. \ref{feedbackmemoryrtclookahead}, with
memoryless source $S=\{U\}_{i\in\mathbf{N}}$ transformed to a Markov source,
$S_M=\{V\}_{i\in\mathbf{N}}$.}
\label{feedbackmemoryrtc}
\end{center}
\end{figure}

\subsection{Average Cost Optimality Equation}
\label{rtccompleteacoe}
\begin{definition}[Bayes Envelope and Bayes Response]\label{def1}
Consider a random variable $X$ taking values in a finite alphabet $\mathcal{X}$
with distribution $\sim P_X$ and $\hat{x}\in\hat\mathcal{X}$ is our guess. The
loss function $\Lambda :
\mathcal{X}\times\hat{\mathcal{X}}\rightarrow\mathbf{R}$ can be
understood as quantifying the discrepancy in the actual value of $X$ and its
estimate. An estimate is good if its expected loss $\E[\Lambda(X,\hat{X})]$ is
small. We define the \textit{Bayes Envelope} as
$B(P_X)=\min_{\hat{x}}\E_{P_X}[\Lambda(X,\hat{x})]$. This represents
the minimal expected loss value associated with the best guess possible. The
best guess is called the \textit{Bayes Response} to $P_X$ and is denoted as
$\hat{X}_{Bayes}(P_X)=\argmin_{\hat{x}}\E[\Lambda(X,\hat{x})]$, where ties are
resolved arbitrarily. In the presence
of observation, the optimal estimator of X based on Y in the sense of
minimizing expected loss under $\Lambda$ is given by
$\hat{X}_{Bayes}(P_{X|Y})=\argmin_{\hat{x}}\E[\Lambda(X,\hat{x})|Y]$.
Note that in general, the \textit{Bayes} response depends on the loss
function, this dependence is implied whenever we use \textit{Bayes} response.
\end{definition}

\begin{lemma}\label{lemma1}
 The optimal decoding rule for the problem in Fig. \ref{feedbackmemoryrtc}
is given by,
\bea
\hat{V}^{opt}_{i}(Y^i)=\hat{V}_{Bayes}(P_{V_i|Y^{i}}).
\eea
\end{lemma}
\begin{proof}
 Fix $n$ and the encoding rule. From the definition of \textit{Bayes} response,
\bea
\E\left[\tilde{\Lambda}(V_i,\hat{V}_i(y^i))\Big{|}Y^i=y^i\right]
&\ge&\E\left[\tilde{\Lambda}(V_i,\hat{V}_{Bayes}(P_{V_i|y^i}))\Big{
|}Y^i=y^i\right ] ,
\eea
which implies,
\bea
\E\left[\tilde{\Lambda}(V_i,\hat{V}_i(Y^i))\right]
&\ge&\E\left[\tilde{\Lambda}(V_i,\hat{V}_{Bayes}(P_{V_i|Y^i}))\right].
\eea
Thus we have the following lower bound on the expected average cost,
\bea
\limsup_{n\rightarrow\infty}\frac{1}{n}\E\left[\sum_{i=1}^{n}\tilde{\Lambda}(V_i
, \hat { V }
_i(Y^i))\right ] \ge\limsup_{n\rightarrow\infty}\frac { 1 } {
n}\E\left[\sum_{i=1}^{n}\tilde{\Lambda}(V_i,\hat{V}_{Bayes}(P_{V_i|Y^i} ))\right
],
\eea
which is attained by decoding rule
$\hat{V}^{opt}_{i}(Y^i)=\hat{V}_{Bayes}(P_{V_i|Y^{i}})$. Thus the optimal
decoding for the
original source, $\hat{U}^{opt}_{i}(Y^i)=\hat{U}_{Bayes}(P_{U_i|Y^{i}})$.
\end{proof}

Fix the decoding rule to be the optimal rule
$\{\hat{V}^{opt}_i\}_{i\in\mathbf{N}}$ as above. Consider the state sequence for
this problem,
$S_i=(V_i,P_{V_i|Y^i})\in\mathcal{S}$.
$P_{V_i|Y^i}$ denotes the belief of the encoder on the
source symbol given all the past and the present channel outputs. Let us denote
it by a $\card{\mathcal{V}}$-dimensional non-negative probability
(column) vector 
$\beta_{Y^i}$. As source symbols
takes values in a finite alphabet, the state space $\mathcal{S}$ is a compact
subset of
Borel space. Consider the disturbance to
be $W_i=(V_i,Y_i)$, which takes values in a finite set,
$\mathcal{V}\times\mathcal{Y}$. The action is history dependent,
$A_i=f_{e,i}(S_0,V^{i-1},Y^{i-1})=f_{e,i}(S_0,W^{i-1})$ (here $S_0$ is some
fixed
initial 
state). $P_{S}$ is some initial distribution. From now on we will use
$f_{e,i}(V^{i-1},Y^{i-1})$ interchangeably with
$f(S_0,V^{i-1},Y^{i-1})$ to denote $A_i$ as $S_0$ is fixed. The
action set is the set of mappings from $\mathcal{V}$ to $\mathcal{X}$, hence
$\card{\mathcal{A}}=\card{\mathcal{X}}^{\card{\mathcal{V}}}$, which is finite.
Note,
\bea
P(W_i|S^{i-1},A^i)&=&P(V_i,Y_i|V^{i-1},A^i,\beta_{Y_1},\cdots,\beta_{Y^{
i-1}})\\
&\stackrel{(\ast)}{=}&K(V_{i-1},V_i)P(Y_i|X_i=A_i(V_i))\label{W11}\\
&=&P_W(W_i|S_{i-1},A_i),\label{W1}
\eea
where ($\ast$) follows from the fact that $\{V_i\}_{i\in\mathbf{N}}$ is Markov,
and from the DMC property of the channel as in Eq. (\ref{DMC}). Hence $W_i\sim
P_W(\cdot|S_{i-1},A_i)$ where $P_W$ is given by Eqs. (\ref{W11}) and (\ref{W1}).

\begin{lemma}\label{lemma2}
Given knowledge of the entire past history of actions, states and disturbance,
the current state evolves according to a deterministic function of the past
state, current action and the current disturbance, i.e.,
\bea
S_i=F(S_{i-1},A_i,W_i).
\eea
\end{lemma}
\begin{proof}
\bea
P_{V_i=v|Y^i}&=&\frac{P_{v,Y_i|Y^{i-1}}}{P_{Y_i|Y_{i-1}}}\\
&=&\frac{\sum_{V_{i-1}}P_{V_{i-1},v,Y_i|Y^{i-1}}}{\sum_{V_{i-1},V_i}P_{V_{i-1}
,V_i,Y_i|Y_{i-1}}}\\
&=&\frac{\sum_{V_{i-1}}P_{V_{i-1}|Y^{i-1}}
K(V_{i-1},v)P(Y_i|X_i=A_i(v))}{\sum_{V_{i-1},V_i}P_{V_{i-1}|Y^{i-1}}
K(V_{i-1},V_i)P(Y_i|X_i=A_i(V_i))}.
\eea 
Therefore,
\bea
\beta_{Y^i}&=&\left[P_{V_i=v|Y^i}\right]_{v\in\mathcal{V}}\\
&=&\left[\frac{\beta_{Y^{i-1}}^{T}\mathbf{K}(v)P(Y_i|A_i(v))}{\sum_{\tilde{v}
\in\mathcal{V}}\beta_{Y^{i-1}}^{T}\mathbf{K}(\tilde{v})P(Y_i|A_i(\tilde{v}))}
\right]_{v\in\mathcal{V}}\\
&=&G(\beta_{Y^{i-1}},A_i,Y_i)\label{G},
\eea
where $\mathbf{K}(v)=[K(\tilde{v},v)]_{\tilde{v}\in\mathcal{V}}$ is a column
vector. Since $W_i=(V_i,Y_i)$, $S_i=(V_i,\beta_{Y^i})$ Eq. (\ref{G}) implies,
\bea
S_i=(V_i,\beta_{Y^i})&=&(V_i,G(\beta_{Y^{i-1}},A_i,Y_i))\\
&=&G'(V_i,Y_i,A_i,\beta_{Y^{i-1}})\\
&=&F(S_{i-1},A_i,W_i)\label{F1}.
\eea
\end{proof}
Let,
\bea
g(S_i,A_{i+1})&=&c(S_i)\\
&=&c(\beta_{Y^i})\\
&=&-\E\left[\tilde{\Lambda}(V_i,\hat{V}_{Bayes}(\beta_{Y^i}))|Y^i\right].
\eea
Therefore,
\bea
\inf\limsup_{n\rightarrow\infty}\frac{1}{n}\E\left[\sum_{i=1}^{n}
\tilde{\Lambda}(V_i ,
\hat { V } ^ {
opt}_i(Y^i))\right]=-\sup\liminf_{n\rightarrow\infty}\frac{1}{n}\E\left[\sum_{
i=1 }^{n}g(S_{i-1},A_{i})\right].
\eea
Hence the tuple 
$\mathcal{T}=(\mathcal{S},\mathcal{A},\mathcal{W},F,P_S,P_W,g)$ forms a
controlled Markov process. The problem of finding the best channel encoder
(using the optimal decoder to be the Bayesian
$V_i(Y^i)=V^{opt}_i(P_{V_i|Y^i})$) in our problem of real time
communication is equivalent to the problem
of finding the optimal policy for the tuple $\mathcal{T}$ which maximizes the
average reward under the cost function $g$. The optimal reward is given by,
\bea
\lambda^{opt}_{\mathcal{T}}=\sup_{\pi}\liminf_{n\rightarrow\infty}\frac{1}{n}
\E\left[\sum_{i=1}^{n}g(S_{i-1},A_{i})\right].
\eea
Thus the ACOE for the controlled Markov process 
$\mathcal{T}=(\mathcal{S},\mathcal{A},\mathcal{W},F,P_S,P_W,g)$ which has the
generic form,
\bea
\lambda+h(s)&=&\sup_{a\in\mathcal{A}}\left[g(s,a)+\sum_{w\in\mathcal{W}}
P_W(w|s,a)h(F(s,a,w))\right], \mbox{ }\forall {s\in\mathcal{S}},
\eea
when specialized to our setting becomes, 
\bea
\lambda+h(v,\beta)&=&\sup_{a\in\mathcal{A}}\left[g(v,\beta,a)+\sum_{w\in\mathcal
 {W}}P_W(w|v,\beta,a)h(F(v,\beta,a,w))\right], \mbox{ }\forall
{v\in\mathcal{V},}\mbox{ }\beta\in\mathcal{P}(\mathcal{V}),
\eea
which becomes, upon substitution from Eq. (\ref{W11}), 
\bea
\lambda+h(v,\beta)&=&c(\beta)+\sup_{a\in\mathcal{A}}\left[\sum_{
(\tilde{v},\tilde{y})\in\mathcal
 {V}\times\mathcal{Y}}K(v,\tilde{v})P_{Y|X}(\tilde{y}|a(\tilde{v}))h(\tilde{v},
G(\beta,a,\tilde{y}))\right],\mbox{ }\forall
{v\in\mathcal{V},}\mbox{ }\beta\in\mathcal{P}(\mathcal{V}).
\eea
We will now transform back the setting from Markov source
$\{V_i\}_{i\in\mathbf{N}}$ to i.i.d. source $\{U_i\}_{i\in\mathcal{N}}$. Let
us denote $v=(u_1,u_2,\cdots,u_{d+1})$ and
$\beta = \beta(\hat{u}_1,\cdots,\hat{u}_{d+1})_{(\hat{u}_1,\cdots,\hat{u}_{d+1}
)\in\mathcal{U}^{d+1}}$. Note that,
\bea
\E\left[\tilde{\Lambda}(V_i,\hat{V}^{opt}(Y^i))\Big{|}Y^i\right]&=&\E\left[
\tilde{ \Lambda }
(V_i , \hat{V}_ { Bayes } (P_ { V_i|Y^i }
))\Big{|}Y^i\right]\\
&=&\E\left[\Lambda(U_i,\hat{U}^{opt}(Y^i))\Big{|}Y^i\right]\\
&=&\E\left[\Lambda(U_i,\hat{U}_{Bayes}(P_{
U_i|Y^i } ))\Big{|}Y^i\right]\\
&\stackrel{(\ast)}{=}&\min_{\hat{u}}\sum_{u\in\mathcal{U}}P_{U_i|Y^i}
(u)\Lambda(u,\hat{u}),
\eea
where $(\ast)$ follows from the Definition \ref{def1}. 
Hence, 
\bea\label{g1}
g(s,a)=c(\beta)=-\min_{\hat{u}}\sum_{u\in\mathcal{U}}\beta_1(u)\Lambda(u,\hat{
u } ),
\eea
where
$\beta_1(u)=\sum_{(\hat{u}_2,\cdots,\hat{u}_{d+1})\in\mathcal{U}^{d}}\beta(u,
\hat {u}_2,\cdots,\hat{ u } _ {d+1})$ is the marginal of $\beta$ on the first
component. Note that $g(\cdot)$ is continuous in $\beta$. Thus the transformed
ACOE to our original problem with i.i.d. source,

\begin{algorithm}
\bea
&&\lambda+h(u_1,\cdots,u_{d+1},\beta)\nonumber\\
&=&-\min_{\hat{u}}\sum_{u\in\mathcal{U}}\beta_1(u)\Lambda(u,\hat{
u } )+\max_ {                                         
a\in\mathcal{A
} } \left [ \sum_ {
(\tilde{u},\tilde{y})\in\mathcal
 {U}\times\mathcal{Y}}P_U(\tilde{u})P_{Y|X}(\tilde{y}|a(u_2,\cdots,u_{d+1}
,\tilde { u }
))h(u_2.\cdots,u_{d+1},\tilde{u},G(\beta,a,\tilde{y}
))\right],\nonumber\\&&\mbox{
}\forall
{(u_1,\cdots,u_{d+1})\in\mathcal{U}^{d+1},}\mbox{
}\beta\in\mathcal{P}(\mathcal{U}^{d+1}).\label{B1}
\eea
\end{algorithm}

\begin{note}[Structure of Optimal Policy]
\label{note3}
As $\mathcal{A}$ is finite, we have replaced the \textit{sup} with
\textit{max} in the Eq. (\ref{B1}). Specializing Theorem \ref{theorem1} for
the
above ACOE, if there exists, a constant, $\lambda^{\ast}$ and a measurable real
valued bounded 
function $h:\mathcal{S}\rightarrow\mathbf{R}$ such that equation Eq. (\ref{B1})
is satisfied for all
$(u_1,\cdots,u_{d+1})\in\mathcal{U}^{d+1},\beta\in\mathcal{P}(\mathcal{U}^{d+1}
)$ then
the minimum distortion $D(d)=-\lambda^{opt}=-\lambda^{\ast}$. Further, if there
exists
a
function, $\mu:\mathcal{S}\rightarrow\mathcal A$ such that the maximum for all
states in Eq. (\ref{B1}) is attained by
$\mu(s)=\mu(u^{d+1}_1,\beta)$, then the optimal encoding policy
is \textit{stationary} and depends on history only through the past state,
i.e., 
$\mu_i(\Phi_i)=\mu(\cdot,S_{i-1})=\mu(\cdot,U^{d}_{i-1},\beta_{Y^{i-1}})$, or
the input to the channel at $i^{th}$ time epoch, is
$X_i=A_i(U_i^{d+i})=\mu(U_i^{d+i},U^{d}_{i-1},\beta_{Y^{i-1}})=\mu(U_{i-1},
\cdots,U_ {d+1},\beta_{Y^{i-1}})$. Hence the optimal encoding in this case is a
stationary mapping into $\mathcal{X}$ which uses only $d+2$ source symbols
$U_{i-1}^{d+1}$ and the belief $\beta_{Y^{i-1}}$ that is updated by the Eq.
(\ref{G}).
\end{note}

\subsection{To Look or Not to Lookahead : Optimality of Symbol by Symbol
Policies}
\label{rtccompletelookahead}
In this section we derive conditions for \textit{stationary, symbol by
symbol} policies to be optimal. This means that we seek to identify situations
where the optimal encoding at time $i$ is given by, $X_i=\mu_{symbol}(U_i)$.

\begin{lemma}\label{lemma3}
When lookahead $d=0$, the minimum average
distortion is achieved by symbol-by-symbol encoding (and decoding) and given
by $D_{symbol}=\min_{\hat{X}:\mathcal{U}\rightarrow\mathcal{X}}
\E\left[ \Lambda\left(U , \hat { U } _ { Bayes } (P_ { U|Y } )\right)\right ]$
\end{lemma}
\begin{proof}
Consider the communication system in Fig. \ref{feedbackmemoryrtclookahead} with
lookahead, $d=0$. Thus this corresponds to a communication system with
memoryless source and memoryless channel, and causal encoding and causal
decoding with
unit delay feedback. We will first use standard information theoretic methods
to prove,
\bea\label{dsymbol}
D_{symbol}=\min_{\hat{X}:\mathcal{U}\rightarrow\mathcal{X},\hat{U}:\mathcal{Y}
\rightarrow\hat{\mathcal{U}}}
\E\left[ \Lambda\left(U , \hat { U }(Y )\right)\right ].
\eea
\textit{Achievability} :\\
Let $D_{min}$ denote the minimum distortion. Clearly $D_{symbol}$ is achievable
by encoding, $X(\cdot)$ and decoding, $\hat{U}(\cdot)$ which attain the minimum
in Eq. (\ref{dsymbol}). Hence $D_{min}\le D_{symbol}$.\\
\textit{Converse} :\\
Consider the chain of inequalities to prove $D_{min}\ge D_{symbol}$. Let $D$ be
the distortion achieved by any causal encoding and causal decoding.
Also note that minimizing over functions of the form, $f_{e,i}(U^i)$ and
$f_{d,i}(Y^i)$ is equivalent to minimizing over vector valued mappings of the
form, $f_{e,i}(\cdot,U^{i-1}):\mathcal{U}\rightarrow\mathcal{X}$ and 
$f_{d,i}(\cdot,Y^{i-1}):\mathcal{Y}\rightarrow\hat{\mathcal{U}}$ 
\bea
D&=&\limsup_{n\rightarrow\infty}\frac{1}{n}\sum_{i=1}^{n}\E\left[\Lambda(U_i,
\hat { U } _i)\right ] \\
&=&\limsup_{n\rightarrow\infty}\frac{1}{n}\sum_{i=1}^{n}\E\left[\E\left[
\Lambda(U_i,\hat{U}_i)|U^{i-1}, Y^{i-1}\right  ] \right]\\
&=&\limsup_{n\rightarrow\infty}\frac{1}{n}\sum_{i=1}^{n}\E\left[\E\left[
\Lambda(U_i,\hat{U}_i)|f_{e,i},f_{d, i}, U^ { i-1 } ,
Y^{i-1}\right  ] \right].
\eea
Note that,
\bea
\E\left[\Lambda(U_i,\hat{U}_i)|f_{e,i},f_{d,
i}, U^ { i-1 } ,
Y^{i-1}\right  ]
&=&\sum_{u\in\mathcal{U},y\in\mathcal{Y}}P_U(u)P_{
Y|X } (y|f_ { e , i } (u))\Lambda(u,f_{d,i}(y))\\
&\ge&\min_{X,\hat{U}}\sum_{u\in\mathcal{U},y\in\mathcal{Y}}P_U(u)P_{
Y|X } (y|X(u))\Lambda(u,\hat{U}(y))\\
&=&D_{symbol},
\eea
which implies, $D\ge D_{symbol}$ for all possible achievable distortions, which
implies, $D_{min}\ge D_{symbol}$. What is left is to show that,
\bea
D_{symbol}=\min_{\hat{X}:\mathcal{U}\rightarrow\mathcal{X},\hat{U}
:\mathcal{Y}
\rightarrow\hat{\mathcal{U}}}
\E\left[ \Lambda(U , \hat { U }(Y) )\right ]
=\min_{\hat{X}:\mathcal{U}\rightarrow\mathcal{X}}
\E\left[ \Lambda\left(U , \hat { U } _ { Bayes } (P_ { U|Y } )\right)\right ],
\eea
which is equivalent to showing that for any encoding rule the optimal decoding
rule, $\hat{U}$ is the \textit{Bayes} response $\hat{U}_{Bayes}(P_{U|Y})$ which
follows from the definition of the \textit{Bayes} response.
\end{proof}
The above proof shows that if \textit{stationary symbol by symbol} policy is
optimal for controlled Markov process of Section \ref{rtccompleteacoe}, then
the optimal reward is given by,
\bea\label{lsymbol}
\lambda_{symbol}=-\min_{\hat{X}:\mathcal{U}\rightarrow\mathcal{X}}
\E\left[ \Lambda\left(U , \hat { U } _ { Bayes } (P_ { U|Y } )\right)\right ].
\eea
Note that the joint distribution of (U,Y) on the right hand side of
Eq. (\ref{lsymbol}) and hence the expected loss is dependent on the encoding
rule $\hat{X}$. To simplify the notation we denote,
$\Lambda(U_i,\hat{U}_{Bayes}(\cdot))$
by $\Lambda(U_i,\cdot)$, the \textit{Bayes} response is implied in this
notation. Also, for a given source, $P_U$, channel $P_{Y|X}$, and a
symbol by symbol encoding policy $\mu_s$, i.e., $X_i=\mu_s(U_i)$, let
$f_{\mu_s}\left[P_U,Y\right]$ denote the posterior, $P_{U|Y}$, when source is
distributed as $P_U$, and encoding policy is $\mu_s$ through the channel
$P_{Y|X}$. Note for brevity we omit indicating $P_{Y|X}$ in the argument of
$f_{\mu_s}(\cdot)$ though the posterior depends on channel also.  
Hence if $\mu_s$ is the minimizer in Eq. (\ref{lsymbol}), then
$\lambda_{symbol}$ is given by
$-\E\left[\Lambda\left(U,f_{\mu_s}\left[P_U,
Y\right]\right)\right]$. To state our next result, pertaining to
the optimality of symbol by symbol coding, we introduce another bit of notation.
The evolution of the posterior is through the function $G(\beta, a, y)$, i.e.,
\bea
\tilde{\beta}&=&G(\beta, a, y)\\
&=&\left[\frac{\beta^{T}\mathbf{K}(v)P(y|a(v))}{
\sum_{\tilde{v}
\in\mathcal{V}}\beta^{T}\mathbf{K}(\tilde{v})P(y|a(\tilde{v}))}
\right]_{v\in\mathcal{V}}.
\eea
Also let for a distribution $\beta$, 
$B(\beta)=\min_{\hat{u}\in\mathcal{U}}\sum_{\overline{u}\in\mathcal{U}}
\beta(\overline{u})\Lambda(\overline{u},\hat{u}) $, which is the Bayes Envelope
for the given loss function.
\begin{theorem}\label{theorem3}
Denote the encoding function which achieves the minimum in Eq. (\ref{dsymbol}) 
by $\mu_s$. For the problem setup depicted in Fig.
\ref{feedbackmemoryrtclookahead} with the ACOE in Eq. (\ref{B1}) for a given
positive lookahead $d \geq 1$, \textit{stationary symbol by symbol} policy is
optimal
if the following holds : 
\bea
&&\sum_{(u,y)\in\mathcal{U}\times\mathcal{Y}}P_U(u)P_{Y|X}
(y|\mu_s(u))B(f_{\mu_s}[P_U,y])+\sum_{k=2}^{d+1}\left[\sum_{
y\in\mathcal{Y}} P_ { Y|X }
(y|\mu_s(u_{k}))B(f_{\mu_s}[\beta_k,y])\right]
\nonumber\\
&=&\sum_ {
(\tilde{u},\tilde{y})\in\mathcal{U}\times\mathcal{Y}}P_U(\tilde{u})P_{Y|X}
(\tilde { y}|a^{\ast}(u_2^{d+1},\tilde{
u}))B(\beta_1)\nonumber\\
&+&\sum_ {
(\tilde{u},\tilde{y})\in\mathcal{U}\times\mathcal{Y}}P_U(\tilde{u})P_{Y|X}
(\tilde { y}|a^{\ast}(u_2^{d+1},\tilde{
u}))\sum_{k=2}^{d}\left(\sum_{\hat{y}\in\mathcal{Y}}
P_{Y|X}(\hat{y}|\mu_s(u_{k+1}))B(f_{\mu_s}[\tilde{\beta}_k,\hat{y}
])\right)\nonumber\\
&+&\sum_ {
(\tilde{u},\tilde{y})\in\mathcal{U}\times\mathcal{Y}}P_U(\tilde{u})P_{Y|X}
(\tilde { y}|a^{\ast}(u_2^{d+1},\tilde{
u}))\left(\sum_{\hat{y}\in\mathcal{Y}}
P_{Y|X}(\hat{y}|\mu_s(\tilde{u}))B(f_{\mu_s}[\tilde{\beta}_{d+1},\hat{y}
])\right)
\nonumber
\eea 
\bea
\forall
{(u_1,\cdots,u_{d+1})\in\mathcal{U}^{d+1},}\mbox{
}\beta\in\mathcal{P}(\mathcal{U}^{d+1}),\label{eq1} 
\eea
where $a^{\ast}(\cdot)$ is the minimizer of
the right hand side of the above equation, $\tilde{\beta}=G(\beta,a,\tilde{y})$
and $\beta_k$
and $\tilde{\beta}_k$ denote the marginal of the $k^{th}$ component of $\beta$
and $\tilde{\beta}$, respectively.
\end{theorem}
\begin{proof}\\
We will first assume that symbol by symbol policy is optimal or the optimal
encoding is
$X_i=\mu_i(U_i^{d+i},\Phi_i)=\mu_{s}(U_i^{d+i},S_{i-1})=\mu_{s}
(U_i^{d+i})=\mu_{s}(U_i)$. Hence we can solve for an $h$ that
satisfies the                                      
following equation,
\bea\label{hsolve}
\lambda_{symbol}+h(u_1,\cdots,u_{d+1},\beta)&=&-\min_{\hat{u}}\sum_{u\in\mathcal
{U}}\beta_1(u)\Lambda(u,\hat{
u } )+\sum_ {
(\tilde{u},\tilde{y})\in\mathcal
 {U}\times\mathcal{Y}}P_U(\tilde{u})P_{Y|X}(\tilde{y}|\mu_{s}
(u_2))h(u_2.\cdots ,
u_{d+1},\tilde{u},\tilde{\beta})\nonumber\\&&\mbox{
}\forall
{(u_1,\cdots,u_{d+1})\in\mathcal{U}^{d+1},}\mbox{
}\beta\in\mathcal{P}(\mathcal{U}^{d+1}),\tilde{\beta}=G(\beta,\mu_s,\tilde{y}
).                                                        
\eea

We claim that for a given lookahead, $d$, Eq.
(\ref{hsolve})
is satisfied with,
\bea\label{H1}
h(u_1,\cdots,u_{d+1},\beta)=-B(\beta_1)-\sum_{k=2}^{d+1}\left[\sum_{
\hat{y}\in\mathcal{Y}}P_{Y|X}(\hat{y}|\mu_s(u_k))B(f_{\mu_s}[\beta_k,\hat{y}])
\right ] ,
\eea
where $\beta_k$ is the marginal
$\beta_k(u_k)=\sum_{(u_1^{k-1},u_{k+1}^{d+1})}\beta(u_1,\cdots,u_{d+1})$.
To prove the claim, consider the L.H.S. of Eq. (\ref{hsolve}),
\bea
\mbox{LHS}&=&\lambda_{symbol}+h(u_1,\cdots,u_{d+1},\beta)\\
&=&-\E\left[\Lambda\left(U,f_{\mu_s}\left[P_U,
Y\right]\right)\right]-B(\beta_1)-\sum_{k=2}^{d+1}\left[\sum_{
\hat{y}\in\mathcal{Y}}P_{Y|X}(\hat{y}|\mu_s(u_k))B(f_{\mu_s}[\beta_k,\hat{y}])
\right ] \\
&=&-\sum_{(u,y)\in\mathcal{U}\times\mathcal{Y}}P_U(u)P_{Y|X}(y|\mu_s(u))B(f_{
\mu_ s}[P_U,y])\nonumber\\
&&-B(\beta_1)-\sum_{k=2}^{d+1}\left[\sum_{
\hat{y}\in\mathcal{Y}}P_{Y|X}(\hat{y}|\mu_s(u_k))B(f_{\mu_s}[\beta_k,\hat{y}])
\right ].
\eea
Before evaluating the right hand side of Eq. (\ref{hsolve}), we evaluate the
marginals
$\tilde{\beta}_k$. 
\bea
\tilde{\beta}(\hat{u}_{1}^{d+1})&=&G(\beta,\mu_s,\tilde{y})(\hat{u}_{1}^{d+1})\\
&=&\frac{\beta^{T}\mathbf{K}(\hat{u}_{1}^{d+1})P_{Y|X}(\tilde{y}|\mu_{s}(\hat{u}
_ { 1 } ^ {
d+1}))}{ \sum_ {
\tilde{u}_{1}^{d+1}}\beta^{T}\mathbf{K}(\tilde{u}_{1}^{d+1})P_{Y|X}(\tilde{
y}|\mu_ {s}(\tilde { u }_{1}^{d+1} ))}\\
&=&\frac{\sum_{u}\beta(u,\hat{u}^d_1)P_U(\hat{u}_{d+1})P_{Y|X}(\tilde{y}|\mu_{s}
(\hat{u}^{d+1}_1))}{\sum_{\hat{u}_1^{d+1}}\sum_{u}\beta(u,\hat{u}^d_1)P_U(\hat{u
}_{d+1})P_{Y|X}(\tilde{y}|\mu_{s}(\hat{u}^{d+1}_1))}\\
&=&\frac{\sum_{u}\beta(u,\hat{u}^d_1)P_U(\hat{u}_{d+1})P_{Y|X}(\tilde{y}|
\mu_s(\hat{u}_1))}{\sum_{\hat{u}_1^{d+1}}\sum_{u}\beta(u,\hat{u}^d_1)P_U(\hat{u}
_ { d+1 } )P_{Y|X}(\tilde{y}|\mu_s({u}_{1}))}.
\eea
Hence the marginals,  
\bea 
\tilde{\beta}_k(\hat{u}_k)=\sum_{(\hat{u}_1^{k-1},\hat{u}_{k+1}^{d+1})}\tilde{
\beta}(\hat{u}_{1}^{d+1})&=&f_{\mu_s}[\beta_2,\tilde{y}](\hat{u}_1),
\mbox{ }k=1\\
&=&\beta_{k+1}(\hat { u } _k) , \mbox { } 2\le k\le d,\\
&=&P_U(\hat{u}_{d+1}),\mbox{  }k=d+1.
\eea
Thus RHS of the Eq. (\ref{hsolve}),
\bea
&=&-B(
\beta_1)+\sum_ {
(\tilde{u},\tilde{y})\in\mathcal
 {U}\times\mathcal{Y}}P_U(\tilde{u})P_{Y|X}(\tilde{y}|\mu_s(u_2))h(u_2.\cdots,
u_{d+1},\tilde{u},\tilde{\beta})\\
&=&-B(\beta_1)-\sum_ {
(\tilde{u},\tilde{y})\in\mathcal
 {U}\times\mathcal{Y}}P_U(\tilde{u})P_{Y|X}(\tilde{y}|\mu_s(u_2))B(\tilde{\beta}
_1)\nonumber\\
&&-\sum_ {
(\tilde{u},\tilde{y})\in\mathcal
 {U}\times\mathcal{Y}}P_U(\tilde{u})P_{Y|X}(\tilde{y}|\mu_s(u_2))\sum_{k=2}^{d}
\left[\sum_{
\hat{y}\in\mathcal { Y } } P_ { Y|X }
(\hat{y}|\mu_s(u_{k+1}))B(f_{\mu_s}[\tilde{\beta}_k,\hat{y}])
)\right ]
\nonumber\\   
&&-\sum_ {
(\tilde{u},\tilde{y})\in\mathcal
 {U}\times\mathcal{Y}}P_U(\tilde{u})P_{Y|X}(\tilde{y}
|\mu_s(u_2)\sum_{\hat{y}\in\mathcal{Y}}P_ { Y|X }
(\hat{y}|\mu_s(\tilde{u})B(f_{\mu_s}[\tilde{\beta}_{d+1},\hat{y}])\\
&=&-B(\beta_1)-\sum_ {
(\tilde{u},\tilde{y})\in\mathcal
 {U}\times\mathcal{Y}}P_U(\tilde{u})P_{Y|X}(\tilde{y}|\mu_s(u_2))B(f_{\mu_s}[
\beta_2,\tilde{y}])\nonumber\\
&&-\sum_{k=2}^{d}\left[\sum_ {
(\tilde{u},\tilde{y})\in\mathcal
 {U}\times\mathcal{Y}}P_U(\tilde{u})P_{Y|X}(\tilde{y}
|\mu_s(u_2))\sum_{\hat{y}\in\mathcal{Y}}P_ {
Y|X }
(\hat{y}|\mu_s(u_{k+1}))B(f_{\mu_s}[
\beta_{k+1},\hat{y}])\right
] \nonumber\\
&&-\sum_ {
(\tilde{u},\tilde{y})\in\mathcal
 {U}\times\mathcal{Y}}P_U(\tilde{u})P_{Y|X}(\tilde{y}
|\mu_s(u_2)\sum_{\hat{y}\in\mathcal{Y}}P_ { Y|X }
(\hat{y}|\mu_s(\tilde{u})B(f_{\mu_s}[P_U,\hat{y}])
\\
&=&-B(\beta_1)-\sum_ {
\tilde{y}\in\mathcal{Y}}P_{Y|X}(\tilde{y}|\mu_s(u_2))B(f_{\mu_s}[
\beta_2,\tilde{y}])-\sum_{k=2}^{d}\left[\sum_{\hat{y}\in\mathcal{Y}}P_ {
Y|X }
(\hat{y}|\mu_s(u_{k+1}))B(f_{\mu_s}[
\beta_{k+1},\hat{y}])\right
] \nonumber\\
&&-\sum_ {
(\tilde{u},\hat{y})\in\mathcal
 {U}\times\mathcal{Y}}P_U(\tilde{u})P_ { Y|X }
(\hat{y}|\mu_s(\tilde{u}))B(f_{\mu_s}[P_U,\hat{y}])\\
&=&\mbox{ LHS}.
\eea
Thus Eq. (\ref{eq1}) and Eq. (\ref{B1}) imply that $\exists$ a bounded
function, $h$ given by Eq. (\ref{H1}) such that we have,
\bea
\lambda_{symbol}+h(u_1,\cdots,u_{d+1},\beta)&=&-B(
\beta_1)+\max_{a\in\mathcal{A}}\left[\sum_ {
(\tilde{u},\tilde{y})\in\mathcal
 {U}\times\mathcal{Y}}P_U(\tilde{u})P_{Y|X}(\tilde{y}|a(u_2^{d+1},
\tilde{u}))h(u_2.\cdots ,
u_{d+1},\tilde{u},\tilde{\beta})\right]\nonumber\\&&\mbox{
}\forall
{(u_1,\cdots,u_{d+1})\in\mathcal{U}^{d+1},}\mbox{
}\beta\in\mathcal{P}(\mathcal{U}^{d+1}),\tilde{\beta}=G(\beta,\mu_s,\tilde{y}),
\eea
and the maximization is attained by the
policy, $\mu_{s}(\cdot,s)=\mu_{s}(\cdot)$ such that
$\mu_{s}(u_1,\cdots,u_{d+1})=\mu_{s}(u_1)$. This implies by Theorem
\ref{theorem1} that the symbol by symbol coding policy is optimal.\\ 
\end{proof}

\begin{corollary}\label{corollary1}
 Given $\mathcal{X}\supseteq\mathcal{U}$, \textit{uncoded symbol by
symbol} policy, i.e., $\mu_s(U_i)=\mu_{c}(U_i)=U_i$ is optimal if,
\bea
&&\sum_{(u,y)\in\mathcal{U}\times\mathcal{Y}}P_U(u)P_{Y|X}
(y|u)B(f_{\mu_{c}}[P_U,y])+\sum_{k=2}^{d+1}\left[\sum_{
y\in\mathcal{Y}} P_ { Y|X }
(y|u_{k})B(f_{\mu_c}[\beta_k,y])\right]
\nonumber\\
&=&\sum_ {
(\tilde{u},\tilde{y})\in\mathcal{U}\times\mathcal{Y}}P_U(\tilde{u})P_{Y|X}
(\tilde { y}|a^{\ast}(u_2^{d+1},\tilde{
u}))B(\beta_1)\nonumber\\
&+&\sum_ {
(\tilde{u},\tilde{y})\in\mathcal{U}\times\mathcal{Y}}P_U(\tilde{u})P_{Y|X}
(\tilde { y}|a^{\ast}(u_2^{d+1},\tilde{
u}))\sum_{k=2}^{d}\left(\sum_{\hat{y}\in\mathcal{Y}}
P_{Y|X}(\hat{y}|u_{k+1}))B(f_{\mu_c}[\tilde{\beta}_k,\hat{y}
])\right)\nonumber\\
&+&\sum_ {
(\tilde{u},\tilde{y})\in\mathcal{U}\times\mathcal{Y}}P_U(\tilde{u})P_{Y|X}
(\tilde { y}|a^{\ast}(u_2^{d+1},\tilde{
u}))\left(\sum_{\hat{y}\in\mathcal{Y}}
P_{Y|X}(\hat{y}|\tilde{u})B(f_{\mu_c}[\tilde{\beta}_{d+1},\hat{y}
])\right)
\nonumber
\eea 
\bea
\forall
{(u_1,\cdots,u_{d+1})\in\mathcal{U}^{d+1},}\mbox{
}\beta\in\mathcal{P}(\mathcal{U}^{d+1}), 
\eea
where $a^{\ast}(\cdot)$ is the minimizer of
the right hand side of the above equation, $\tilde{\beta}=G(\beta,a,\tilde{y})$
and $\beta_k$
and $\tilde{\beta}_k$ denote the marginal of the $k^{th}$ component of $\beta$
and $\tilde{\beta}$, respectively.
\end{corollary}
\begin{proof}
Substitute in Theorem \ref{theorem3}, $\mu_s=\mu_{c}$.
\end{proof}

\section{Real-Time Coding with Limited Lookahead : Finite Memory}
\label{rtcfinite}
\subsection{Average Cost Optimality Equation}
\label{rtcfiniteacoe}
In Section \ref{rtccomplete}, we considered the scenario where the
decoder has access to the entire past channel output sequence or, equivalently 
memory is unbounded. In this section, we will develop controlled Markov process
formulation for the case where the memory alphabet is finite and does not grow
with time i.e, 
 the memory space ($\mathcal{M}$) is time-independent
with $\card{\mathcal{M}}<\infty$. We make two assumptions on our coding systems
for this setting :
\begin{itemize}
\item[A1] There is a fixed time-independent memory update function, i.e., there
exists a function $f_m$ such that $Z_i = f_m(Z_{i-1}, Y_i)$ for all $i$. This
assumption is not very restrictive as real systems such as quantizers or finite
window storage devices store only the past few channel output symbols and evolve
in a time invariant way, eg. $Z_i=f_{m}(Z_{i-1},Y_i)=Y_i$,
implies the reconstruction is given by, $\hat{U}_i=f_{d,i}(Y_i,Y_{i-1})$.
\item[A2] We fix the optimal decoding
rule to be,
$f_{d,i}(Y_i,Z_{i-1})=\hat{U}^{opt}(Y_i,Z_{i-1})$, that is the decoding 
is restricted to optimal policies among the stationary (time invariant) ones.
Note that though we assume stationary decoding, optimal encoding may in general
not be stationary.
\end{itemize}
Hence the optimal expected average distortion depends on 
$d,f_m,\mathcal{M}$ and hence we denote it by,
$D(d,f_m,\mathcal{M})$ to distinguish it from $D(d)$ of Section
\ref{rtccomplete}.
\par
 Here too we begin by formulating a controlled 
Markov process for the modified source, $\{V_i\}_{i\in\mathbf{N}}$ and then
substitute for the original source. By the assumption A2, the optimal decoding
is stationary, $\hat{V}^{opt}(\cdot)$. Consider the state
sequence, $S_i=(V_i,Z_i)\in\mathcal{S}$($=\mathcal{V}\times\mathcal{Z}$) and the
disturbance sequence,
$W_i=(V_i,Y_i)\in\mathcal{W}$($=\mathcal{V}\times\mathcal{Y}$). The actions can
be history dependent, $A_i=f_{e,i}(S_0,W^{i-1})$, $S_0$ is a fixed initial
state (with some distribution $P_S$). The disturbance
depends on the past sequence of disturbances, states, actions and the current
action only through the past state and current action, i.e.,
\bea
P(W_i|W^{i-1},S^{i-1},A^i)&=&P(V_i,Y_i|V^{i-1},Y^{i-1},A^i)\\
&=&K(V_{i-1},V_i)P_{Y|X}(Y_i|A_i(V_i))\\
&=&P_W(W_i|S_{i-1},A_i),\label{W2}
\eea
and hence  
\bea\label{F2}
S_i=(V_i,Z_i)=(V_i,f_{m}(Y_i,Z_{i-1}))=F(S_{i-1},A_i,W_i).
\eea
Note that for our transformation (as in 
Section \ref{rtccomplete}) the modified cost function is given by,
\bea
\tilde{\Lambda}(V_i,\hat{V}^{opt}(Y_i,Z_{i-1}))=\Lambda(U_i,\hat{U}^{opt}(Y_i,Z_
{ i-1
})).
\eea
Therefore consider,
\bea
&&\E\left[\tilde{\Lambda}(V_i,\hat{V}^{opt}(Y_i,Z_{i-1}))|V^{i-1},Y^{i-1},Z^{i-1
} ,
A^i\right]\\
&=&\sum_{v\in\mathcal{V},y\in\mathcal{Y}}P(V_i=v,Y_i=y|V^{i-1},Y^{i-1},Z^{i-1},
A^i)\tilde{\Lambda}(v,\hat{V}^{opt}(y,Z_{i-1}))\\
&=&\sum_{v\in\mathcal{V},y\in\mathcal{Y}}K(V_{i-1},v)P(y|A_i(v))\tilde{\Lambda}
(v ,\hat{V}^{opt}(y,Z_{i-1}))\\
&=&-g(S_{i-1},A_i).
\eea
Thus,
\bea
&&\frac{1}{n}\sum_{i=1}^{n}\E\left[\tilde{\Lambda}(V_i,\hat{V}^{opt}(Y_i,Z_{i-1}
))\right ] \\
&=&\frac{1}{n}\sum_{i=1}^{n}\E\left[\E\left[\tilde{\Lambda}(V_i,\hat{V}^{opt}
(Y_i , Z_
{ i-1 } ))|V^{i-1},Y^{i-1},Z^{i-1},A^i\right]\right]\\
&=&-\frac{1}{n}\sum_{i=1}^{n}\E\left[g(S_{i-1},A_i)\right],
\eea
and consequently,
\bea
\inf\limsup_{n\rightarrow\infty}\frac{1}{n}\E\left[\sum_{i=1}^{n}\tilde{
\Lambda}(V_i ,
\hat { V } ^ {
opt}(Y^i))\right]=-\sup\liminf_{n\rightarrow\infty}\frac{1}{n}\E\left[\sum_{
i=1 }^{n}g(S_{i-1},A_{i})\right].
\eea
Thus the problem is to find the optimal policy for the controlled markov
process, $\mathcal{T}=(\mathcal{S},\mathcal{A},\mathcal{W},F,P_S,P_W,g)$, which
maximizes the
average reward under the cost function $g$. The optimal reward is given by,
\bea
\lambda^{opt}_{\mathcal{T}}=\sup_{\pi}\liminf_{n\rightarrow\infty}\frac{1}{n}
\E\left[\sum_{i=1}^{n}g(S_{i-1},A_{i})\right].
\eea
Thus the ACOE for the controlled Markov process
$\mathcal{T}=(\mathcal{S},\mathcal{A},\mathcal{W},F,P_S,P_W,g)$, is given by, 
\bea
\lambda+h(s)&=&\sup_{a\in\mathcal{A}}\left[g(s,a)+\sum_{w\in\mathcal{W}}
P_W(w|s,a)h(F(s,a,w))\right], \mbox{ }\forall {s\in\mathcal{S}}\\
\lambda+h(v,z)&=&\sup_{a\in\mathcal{A}}\left[g(v,z,a)+\sum_{w\in\mathcal
 {W}}P_W(w|v,z,a)h(F(v,z,a,w))\right], \mbox{ }\forall
{v\in\mathcal{V},}\mbox{ }z\in\mathcal{Z}.
\eea
We will now transform back the setting from Markov source
$\{V_i\}_{i\in\mathbf{N}}$ to i.i.d. source $\{U_i\}_{i\in\mathcal{U}}$. Let
us denote $v=(u_1,u_2,\cdots,u_{d+1})$. Since
$\tilde{\Lambda}(V_i,\hat{V}^{opt}(Y_i,Z_{i-1}))=\Lambda(U_i,\hat{U}^{opt}(Y_i,
Z_ { i-1 } ))$ ,
we have,
\bea
g(s,a)&=&g(u_1^{d+1},z,a)\nonumber\\
&=&-\sum_{\tilde{u}
\in\mathcal{U}, \tilde{y}
\in\mathcal{Y}}P_U(\tilde{u})P_{Y|X}
(\tilde { y } |a(u_2,\cdots,u_{d+1},\tilde { u } ))\Lambda(u_2 ,
\hat{U}^{opt}(\tilde{y},z)).\label{g2}
\eea
Thus the resulting ACOE for 
our problem of an i.i.d. source with lookahead $d$ (replacing again $\sup$ by
$\max$ due
to finiteness of the action set),
\begin{algorithm}
\bea
&&\lambda+h(u_1,\cdots,u_{d+1},z)\nonumber\\&=&\max_{a\in\mathcal{A
} } \left [ \sum_ {
(\tilde{u},\tilde{y})\in\mathcal
 {U}\times\mathcal{Y}}P_U(\tilde{u})P_{Y|X}(\tilde{y}|a(u_2,\cdots,u_{d+1}
,\tilde { u }
))\left\lbrace h(u_2.\cdots,u_{d+1},\tilde{u},f_{m}(\tilde{y},z))-\Lambda(u_2 ,
\hat{U}^{opt}(\tilde{y},z))\right\rbrace\right],\nonumber\\&&\mbox{
}\forall
{(u_1,\cdots,u_{d+1})\in\mathcal{U}^{d+1},}\mbox{
}z\in\mathcal{Z}.\label{B2}
\eea
\end{algorithm}
\\
Here again, invoking Theorem \ref{theorem1} implies that if the ACOE in Eq.
(\ref{B2}) is
solved by a real $\lambda^{\ast}$ and a bounded $h(\cdot)$, then
$D(d,f_m,\mathcal{M})=-\lambda^{\ast}$.
\begin{note}\label{note4}
The reason for making assumptions A1 and A2, is that the modification of the
cost
function 
$\Lambda(\cdot,\cdot)$ results in $g(\cdot,\cdot)$ which is time invariant 
and a state sequence evolving through a function
$F(\cdot,\cdot,\cdot)$ which is also time invariant.
\end{note}

\begin{theorem}[Optimality of Stationary Policy]
\label{theorem4}
The ACOE Eq. (\ref{B2}) admits a stationary optimal policy.
\end{theorem}
\begin{proof}
This follows from Theorem 4.3 in \cite{Borkar_survey} as for a fixed lookahead,
$d$, the state space and action space are finite. Hence the optimal encoding
is, $X_i=\mu(U_{i-1}^{d+1},Z_{i-1})$.
\end{proof}

\subsection{Computing $D(d,f_m,\mathcal{M})$, Bounding $D(d)$}
\label{rtcfinitecomputation}
In this section, we explicitly compute $D(d,f_m,\mathcal{M})$. Note
that in the setting of complete memory in Section \ref{rtccomplete} the average
cost optimality equation can also be solved approximately. As the state space
is compact, this admits discretization of the space and running value or policy
iteration to obtain approximations to the optimal distortion.
References \cite{Dpapprox1}
and \cite{DPapprox2} provide an extensive treatment along with
prescriptions of error bounds
and the trade off
between quantization resolution and the precision of the approximated optimal
reward for discounted cost
problems. However, the computational point of view we take in this section does
not follow the path of discretization and then approximation of the average
reward.
Rather, we compute exactly $D(d,f_m,\card{\mathcal{M}})$ which provide
non-trivial upper bounds on $D(d)$. This is illustrated by the following
example. \\
We assume $\mathcal{U}=\mathcal{X}=\mathcal{Y}=\hat\mathcal{U}=\{0,1\}$.  The
source is $Bern(p), p\in[0,0.5]$ and the channel is $BSC(\delta),
\delta\in[0,0.5]$ and loss function is hamming distortion. The memory is of $m$
bits and retains the last $m$ channel
outputs and hence, $\hat{U}_i=\hat{U}^{opt}(Y^i_{i-m}$). We will denote the
optimal expected average distortion by $D(d,m)$ in this case. We observe the
following,
\begin{itemize}
 \item $D(0,m)=D(0)=D_{symbol}=\min_{\hat{X}:\mathcal{U}\rightarrow\mathcal{X}}
\E\left[ \Lambda(U , \hat { U } _ { Bayes } (P_ { U|Y } )\right ]= \min \{p,
\delta \}$, $\forall$
$m$.
\item $D(d)\le D(d,m_1)\le D(d,m_2)\le D(0)$ $\forall$ $m_1\ge m_2$.
\item $D(d)\le D(d_1,m)\le D(d_2,m)\le D(0)$ $\forall$ $d_1\ge d_2$.
\item $D(\infty)\le D(d) \le D(d,m)\le D(0)$ $\forall$ $m,d$.
\end{itemize}

\begin{lemma}
 \label{lemma4}
$D(\infty)=D_{min}$ where $D_{min}$ is the minimum achievable distortion of the
joint source-channel communication problem.
\end{lemma}
\begin{proof}
Any sequential or limited delay 
encoding and decoding scheme can obviously be embedded into and emulated
arbitrarily closely by a sequence of block codes. Hence 
$D_{min}\le D(\infty)$. We will now prove $D_{min}\ge D(\infty)$. This is
equivalent to
proving that any sequential scheme with infinite lookahead can be used to
construct a block coding scheme which can then
attain arbitrarily close to $D_{min}$ hence, $D(\infty)\le D_{min}$. The
argument is based on block Markov coding, (cf. Section on
\textit{Coherent Multihop Lower Bound}, Chapter 17, \cite{GamalKim}) except that
instead of
coding
via looking
at the past block as in block Markov coding, here we look at the future block.
Fix an arbitrarily small $\epsilon>0$, an arbitrarily large $B$, and an $n$
sufficiently large that there exists a block coding scheme of blocklength $n$
achieving per symbol distortion no larger than $D_{min}+\epsilon$. Let $f_e$
and $f_d$ denote the encoding and decoding mappings of that $\epsilon$-achieving
scheme. We now construct a sequential scheme with lookahead $d=2n$ as follows : 
\begin{itemize}
 \item \textit{Encoding} : We code in the present block using the source
symbols of the future block, i.e, $X^n(b)=f_e(U_{bn+1}^{(b+1)n})$ for
$b=1,\cdots,B-1$, and some dummy coding for the last block known to the
decoder. 
\item \textit{Decoding} : For the block one decoder has some predefined
construction. For block $b$, decoder constructs source symbols as,
$\hat{U}^n(b)=f_d(Y^{(b-1)n}_{(b-2)n+1})$, $b=2,\cdots,B$, thus decoding is in
blocks using the past block.
\end{itemize}
The per-symbol distortion achieved by this scheme is clearly upper bounded by
$\frac{\Lambda_{max} + (B-1) (D_{min} + \epsilon)}{B}$, which can be made
arbitrarily
close to $D_{\min}$ for sufficiently small $\epsilon$ and sufficiently large
$B$. 
\end{proof}

\begin{figure}[htbp]
\begin{center}
\includegraphics[scale=0.35]{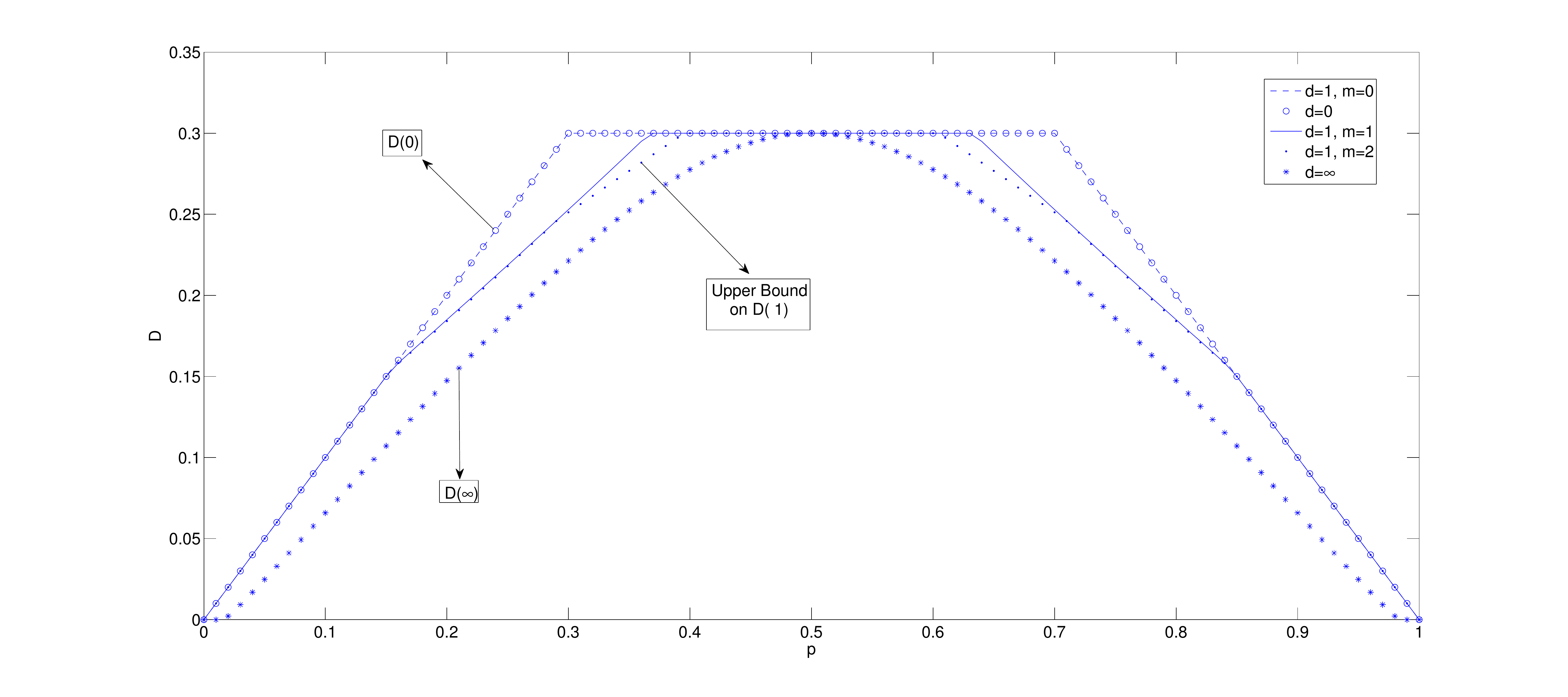}
\caption{Computing and contrasting $D(d,m)$ and $D(d)$ for a
Bernoulli source and binary symmetric channel and Hamming loss. We fix channel
cross over
probability $\delta=0.3$ and vary source probability in $[0,0.5]$. For $d=1$, we
have plotted values for increasing memory, $m=0,1,2$, which
yield series of non-trivial non-increasing upper bounds on $D(1)$. $D(0)$ is
achieved by symbol symbol by
symbol coding, while $D(\infty) = D_{min}$ of Shannon's joint source channel
coding (achieved by separation).}\label{values_p}
\end{center}
\end{figure}

We have run relative value iteration to compute $D(d,m)$ for $d=1$ and some
values of $m$, yielding some interesting upper bounds on $D(1)$. Note that the
values obtained are exact and do not approximate distortion as the relative
value iteration converges in a few iterations. This is because the
state space and action space is finite, and it is easy to check that the
weak accessibility condition (Definition 4.2.2 \cite{DP2}) is satisfied.
This implies by Proposition 4.3.1 of \cite{DP2}, that relative value
iteration converges. Fig.
\ref{values_p} shows the distortion values as a function of source
distribution when the cross over probability is fixed, $\delta=0.3$. Fig.
\ref{values_delta} shows the distortion values as a function of channel cross
over probabilities when the source distribution is fixed, $p=0.3$. 

\begin{figure}[htbp]
\begin{center}
\includegraphics[scale=0.35]{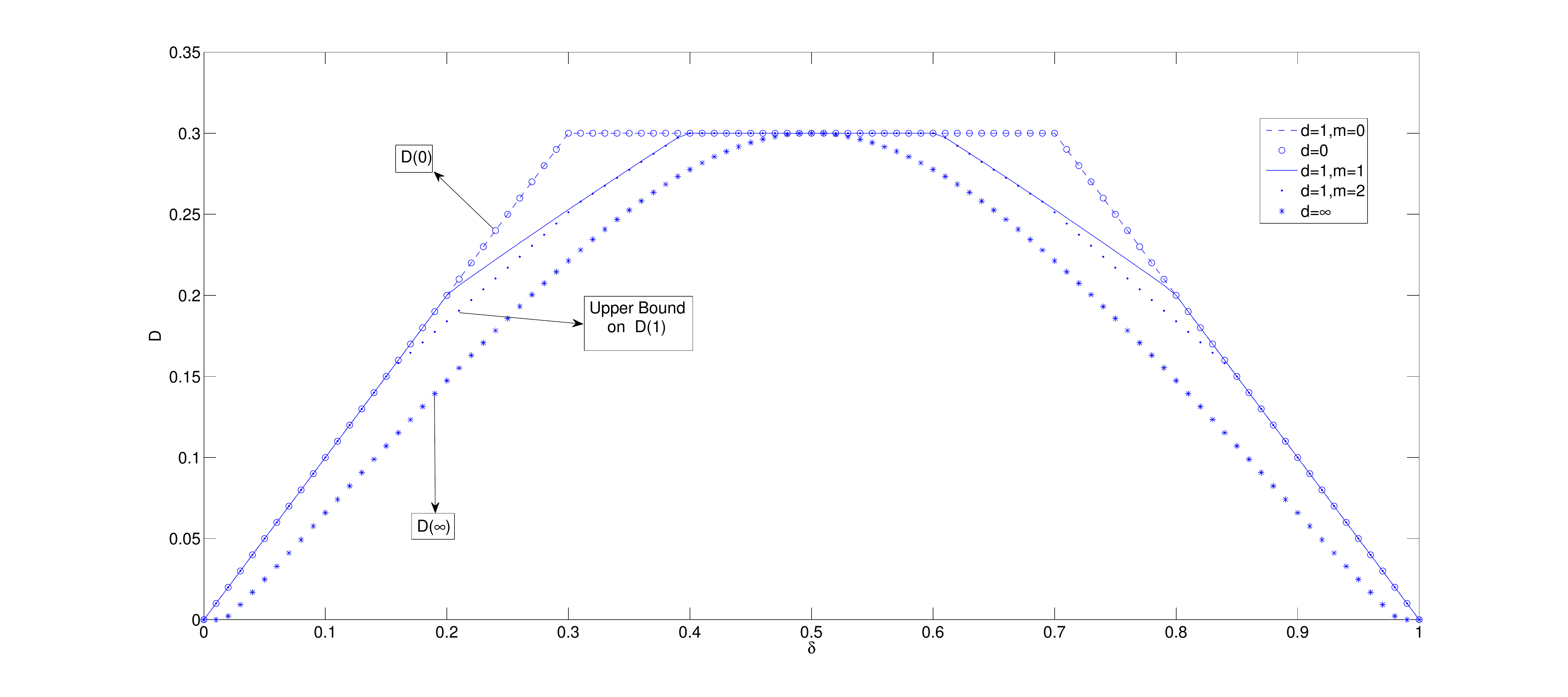}
\caption{ Computing and contrasting $D(d,m)$ and $D(d)$ for a
Bernoulli source and binary symmetric channel and Hamming loss. We fix source
probability $p=0.3$ and vary the channel cross-over probability in
$[0,0.5]$. For $d=1$, we
have plotted values for increasing memory, $m=0,1,2$, which
yield series of non-trivial non-increasing upper bounds on $D(1)$. $D(0)$ is
achieved by symbol symbol by
symbol coding, while $D(\infty) = D_{min}$ of Shannon's joint source channel
coding (achieved by separation).
}\label{values_delta}
\end{center}
\end{figure}

These plots provide insight into the structure of
optimal policies in the setting of Section \ref{rtccomplete}, given that we are
considering $Bern(p)$ source, $BSC(\delta)$ channel under Hamming loss and
$\card{\mathcal{X}}=\card{\mathcal{U}}=2$.  Since $D(1,2)$ is an upper bound on
$D(1)$ and hence on $D(d), d\ge 1$, it is clear that for source
distributions and channel cross over probabilities where $D(1,2) < D(0)$, symbol
by symbol is not
optimal. We evaluate this region and show it in Fig. \ref{eyebddrefined}.
Note when $d=\infty$, since separation is optimal, the region of
suboptimality of symbol by symbol is the the complete square
$(p\in(0,0.5),\delta\in(0,0.5))$ except the boundary where symbol by symbol
is optimal. Also note, as is consistent with the plots, that in the zero
lookahead case, we have
$D(0)=\min \{ p, \delta \}$. Hence, for any lookahead, $d$, for a fixed cross
over probability $\delta$, if symbol by symbol encoding-decoding achieves $D(d)$
for
$p=p_0<0.5$, then it is also optimal for $p\in(p_0,1-p_0]$. Similarly, for a
fixed source probability $p$, if symbol by symbol encoding-decoding is optimal
for $\delta=\delta_0<0.5$, then it is also optimal for
$\delta\in(\delta_0,1-\delta_0]$. 

\begin{figure}[htbp]
\begin{center}
\includegraphics[scale=0.35]{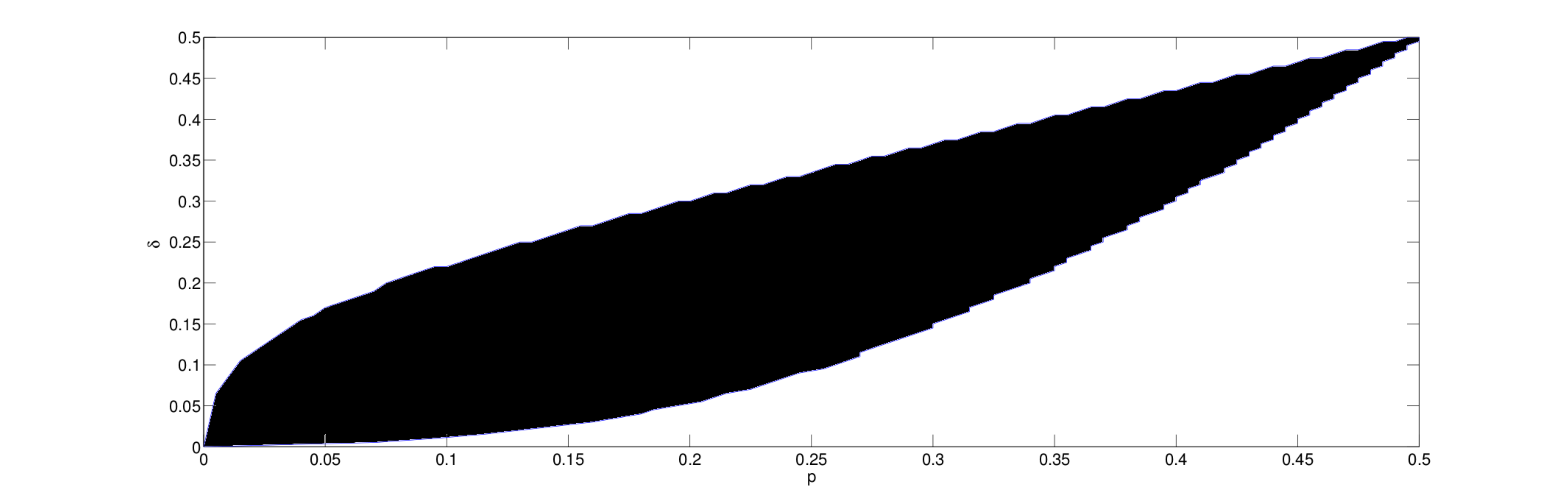}
\caption{The plot shows a region in the source-channel plane where symbol by
symbol coding is suboptimal among schemes with lookahead $d=1$, i.e., it does
not achieve $D(1)$. This is the shaded region which is bounded inside by the two
curves.
}\label{eyebddrefined}
\end{center}
\end{figure}
\newpage 
\section{Real Time Coding with Limited Lookahead : In the Absence of Feedback}
\label{nofeedbackrtc}
In the previous sections we assumed the availability of perfect unit delay
feedback from the decoder to the encoder. We now consider the same setting as
that depicted  in Fig. \ref{feedbackmemoryrtclookahead}, but without feedback,
and we formulate the problem as a
controlled Markov process. Decoders have finite state (see Note \ref{note5}) and
the assumptions A1 and A2 are presumed for similar reasons as in Note
\ref{note4} in Section \ref{rtcfinite}, i.e., memory is finite and decoding is
stationary. Here again, we first study the
system
with modified source, $\{V_i\}_{i\in\mathbf{N}}$. The state space for this 
problem is, $S_i=(V_i,P_{Z_{i}|V^i})\in\mathcal{S}$ and the disturbance
$W_i=V_i$. The actions are thus history dependent,
$A_i=f_{e,i}(S_0,V^{i-1})=f_{e,i}(S_0,W^{i-1})$, $S_0$ is some fixed initial
state
with distribution $P_S$.
Due to Markovity
of the source we have,
\bea
P(W_i|W^{i-1},S^{i-1},A^i)&=&K(V_{i-1},V_i)
\\&=&P_W(W_i|S_{i-1},A_i).\label{W3}
\eea
Denoting $\{P_{Z_i|V^i}\}_{Z_i\in\mathcal{Z}}$ by $\beta_i$,
\bea
\beta_i(z)&=&\frac{P_{V_i,z|V^{i-1}}}{P_{V_i|V_{i-1}}}\\
&=&\frac{\sum_{Z_{i-1}}P_{Z_{i-1},V_i,z|V^{i-1}}}{\sum_{Z_{i-1},Z_i}P_{Z_{i-1}
,V_i,Z_i|V^{i-1}}}\\
&=&\frac{\sum_{Z_{i-1},Y_i}P_{Z_{i-1}|V^{i-1}}
K(V_{i-1},V_i)P(Y_i|A_i(V_i))\1_{\{z=f_m(Y_i,Z_{i-1})\}}}
{\sum_{Z_{i-1},Y_i,Z_i}P_{Z_{i-1}|V^{i-1}}
K(V_{i-1},V_i)P(Y_i|A_i(V_i))\1_{\{Z_i=f_m(Y_i,Z_{i-1})\}}}\\
&=&\frac{\sum_{Z_{i-1},Y_i}\beta_{i-1}(Z_{i-1})
K(V_{i-1},V_i)P(Y_i|A_i(V_i))\1_{\{z=f_m(Y_i,Z_{i-1})\}}}{\sum_{Z_{i-1},Y_i,Z_i}\beta_{i-1}(Z_{i-1})
K(V_{i-1},V_i)P(Y_i|A_i(V_i))\1_{\{Z_i=f_m(Y_i,Z_{i-1})\}}}.
\eea 
This implies,
\bea\label{F3} 
\beta_i=\xi(\beta_{i-1},V_i,V_{i-1},A_i)=\xi(S_{i-1},A_i,W_i), \mbox{
}S_i=F(S_{i-1},A_i,W_i).
\eea
 We now have for average cost, with modified cost function
$\tilde{\Lambda}(V_i,\hat{V}_i)=\Lambda(U_i,\hat{U}_i)$ and stationary decoding
$\hat{V}^{opt}(\cdot)$,
\bea
&&\E\left[\tilde{\Lambda}(V_i,\hat{V}^{opt}(Y_i,Z_{i-1}))|V^{i-1},\beta^{i-1},
A^i\right ]
\\
&=&\sum_{\tilde{z}\in\mathcal{Z},\tilde{v}\in\mathcal{V},\tilde{y}
\in\mathcal{Y}}P(Z_{i-1}=\tilde{z},V_i=\tilde{v},Y_i=\tilde{y} |V^ {i-1} ,
\beta^{i-1}, A^i )\tilde{\Lambda}(\tilde{v} ,
\hat{V}^{opt}(\tilde{y},\tilde{z})))\\
&=&\sum_{\tilde{z}\in\mathcal{Z},\tilde{v}\in\mathcal{V},\tilde{y}
\in\mathcal{Y}}
\beta_{i-1}(\tilde{z})K(V_{i-1},\tilde{v})P(\tilde{y}|A_i(\tilde{v}))\tilde{
\Lambda}(\tilde{v} , \hat{V}^{opt}(\tilde{y},\tilde{z})))\\
&=&-g(V_{i-1},\beta_{i-1},A_i)\\
&=&-g(S_{i-1},A_i).
\eea
Thus,
\bea
&&\frac{1}{n}\sum_{i=1}^{n}\E\left[\tilde{\Lambda}(V_i,\hat{V}^{opt}(Y_i,Z_{i-1}
))\right ] \\
&=&\frac{1}{n}\sum_{i=1}^{n}\E\left[\E\left[\tilde{\Lambda}(V_i,\hat{V}^{opt}(Y_i,Z_
{ i-1 }
))|V^{i-1},\beta^{i-1},A^i\right]\right]\\
&=&-\frac{1}{n}\sum_{i=1}^{n}\E\left[g(S_{i-1},A_i)\right].
\eea
We finally write down the ACOE after transforming to the original source
(similar to the previous sections),
\begin{algorithm}
\bea
\lambda+h(u_1,\cdots,u_{d+1},\beta)&=&\max_{a\in\mathcal{A
} } \left [ g(u_1^{d+1},\beta,a)+\sum_ {
(\tilde{u},\tilde{y})\in\mathcal
 {U}\times\mathcal{Y}}P_U(\tilde{u})P_{Y|X}(\tilde{y}|a(u_2,\cdots,u_{d+1}
,\tilde { u }
))h(u_2.\cdots,u_{d+1},\tilde{u},\tilde{\beta})\right],\nonumber\\&&\mbox{
}\forall
{(u_1,\cdots,u_{d+1})\in\mathcal{U}^{d+1},}\mbox{
}\beta\in\mathcal{P}(\mathcal{Z}).\label {B3}
\eea
\end{algorithm}
\\ where
\bea\label{g3}
g(s,a)=g(u_1^{d+1},\beta,a)
&=&-\sum_{\tilde{z}\in\mathcal{Z},\tilde{u}\in\mathcal{U},\tilde{y}
\in\mathcal{Y}}
\beta(\tilde{z})P_U(\tilde{u})P(\tilde{y}|a(u_2,\cdots,
 u_{d+1},\tilde{u}))\Lambda(u_2,\hat{U}^{opt}(\tilde{y},\tilde{z})),
\eea
and
\bea
\tilde{\beta}&=&\xi((u_1,\cdots,u_{d+1},\beta),a,(u_2,\cdots,u_{d+1},\tilde{u}))\\
&=&\left\lbrace\frac{\sum_{\tilde{z}\in\mathcal{Z},\tilde{y}\in\mathcal{Y}}\beta(\tilde{z})
P(\tilde{u})P(\tilde{y}|a(u_2,\cdots,u_{d+1},\tilde{u}))\1_{\{z=f_m(\tilde{y},\tilde{z})\}}}{\sum_{\tilde{z}\in\mathcal{Z},\tilde{y}\in\mathcal{Y},z\in\mathcal{Z}}\beta(\tilde{z})
P(\tilde{u})P(\tilde{y}|a(u_2,\cdots,u_{d+1},\tilde{u}))\1_{\{z=f_m(\tilde{y},
\tilde{z})\}}}\right\rbrace_{z\in\mathcal{Z}}.
\eea
Theorem \ref{theorem1} implies that if ACOE Eq. (\ref{B3}) is
solved by a real $\lambda^{\ast}$ and a bounded $h(\cdot)$, then
$D(d,f_m,\mathcal{M})=-\lambda^{\ast}$. The results on the structure of optimal
policies parallel those outline in Note \ref{note3} and hence are omitted.
\begin{note}
 \label{note5}
For the setting considered in this section when no feedback is
present, we have restricted our attention to finite state decoders only, unlike 
the previous section where feedback was present and we also considered the case
where decoding used complete memory. This is because in the absence of
feedback when decoding uses complete memory, the state space is one on the
simplex of distributions on alphabets that grows exponentially with the time
index and hence the results of the theory presented in Section
\ref{controlledmarkov} are not as directly applicable.
\end{note}

\section{Sequential Source Coding with A Side Information \textquotedblleft
Vending Machine\textquotedblright}
\label{rtcaction}
In previous sections we considered the problem of real time source-channel
communication when the encoder generates channel input symbol sequentially with
a lookahead, with or without unit delay noise-free feedback, and the decoder
generates the estimate of
the source given the channel output and the memory. In this section we consider
a rate-distortion problem, where encoding is sequential with lookahead. In
addition to it, the decoder can take cost constrained actions, also in a
sequential
fashion, which affect the quality of the side information correlated with the
source
symbol it attempts to reconstruct. We consider two classes of such models : 
one where the encoder has
access to the past side information symbols through unit delay noise-free
feedback (Section \ref{sirtcaction}) and the other when it does not (Section
\ref{nosirtcaction}). The findings of this section are similar in spirit to
those of previous sections and assert the universality of the methodology
invoked
in the paper. We defer the proofs in this section to the Appendix.
\begin{figure}[htbp]
\begin{center}
\scalebox{0.65}{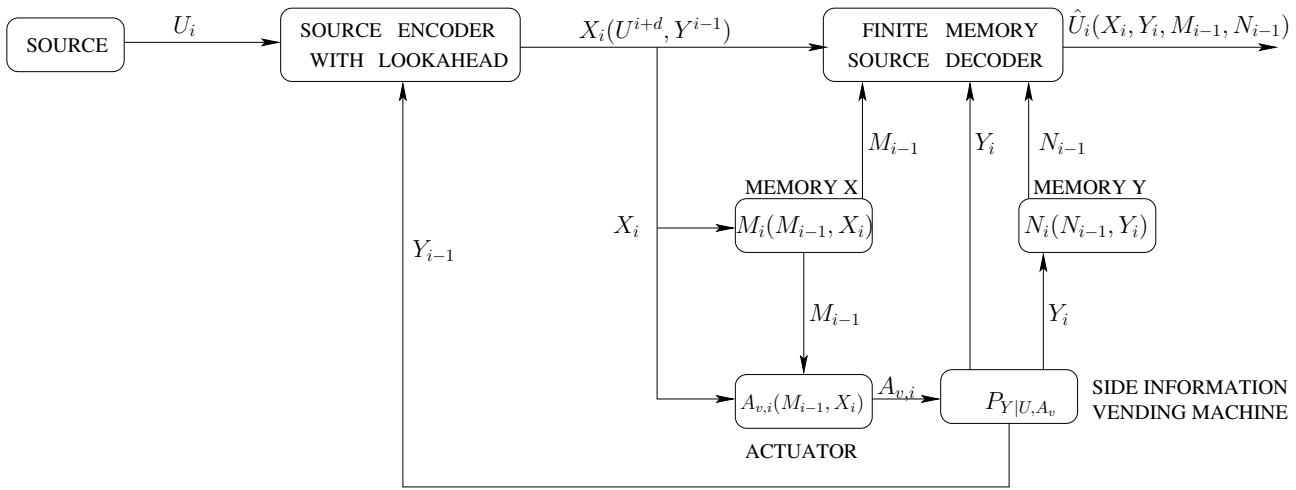}
\caption{The setting of sequential source coding with source lookahead at the
encoder and side
information vending machine at the decoder. The encoder also knows the past
side information symbols through a unit delay noise-free feedback from the
decoder.}
\label{actionfeedbackmemoryrtclookahead}
\end{center}
\end{figure}

\subsection{Encoder has access to Side Information} 
\label{sirtcaction}
The setting depicted in Fig. \ref{actionfeedbackmemoryrtclookahead} consists of
the following
blocks : 
\begin{itemize}
\item \textit{Source Encoder} : The encoder has access to source symbols upto a
lookahead, 
$d$  and to the past side information symbols, i.e,
$X_i=f_{e,i}(U^{i+d},Y^{i-1})$, where $f_{e,i}$ is the encoding
function, $f_{e,i}:\mathcal{U}^{i+d}\times \mathcal{Y}^{i-1}\rightarrow
\mathcal{X}$, $i\in\mathbf{N}$. 
\item \textit{Memory X} : The decoder might not be able to use all of the
encoded symbols upto current time due to memory constraints. Memory X is updated
as a
function of the past state of the memory and the current encoder output,
i.e.,
$M_i=f_{m,i}(M_{i-1},X_i)$, where the $f_{m,i}$ is the memory update
function, $f_{m,i}:\mathcal{M}_{i-1}\times\mathcal{X}\rightarrow\mathcal{M}_{i}$,
$i\in\mathbf{N}$. Note that the alphabet $\mathcal{M}_i$ can grow with $i$,
hence this includes the special case of complete memory, i.e.,
$M_i=X^i$.
\item \textit{Actuator}  : Actuator uses the past Memory X, and 
the current encoded symbol to generate an action, i.e.,
$A_{v,i}=f_{v,i}(M_{i-1},X_i)$, where
$f_{v,i}:\mathcal{M}_{i-1}\times\mathcal{X}_i\rightarrow\mathcal{A}_v$. The
action sequence should satisfy the following cost constraint,
\bea
\limsup_{N\rightarrow\infty}\E\left[\frac{1}{N}\sum_{i=1}^{N}C(A_{v,i})\right]
\le \Gamma,
\eea
where $C(\cdot)$ is the cost function and $\Gamma$ is the cost constraint.
\item \textit{Side Information \textquotedblleft Vending
Machine\textquotedblright} : The side information is 
generated according to $P_{Y|U,A_v}$, i.e.,
\bea
P(y_i|u_1^{\infty},x^i,a_v^i)=P_{Y|U, A_v}(y_i|u_i,a_{v,i}).
\eea
\item \textit{Memory Y} : The decoder  may be limited in its ability to remember
all the side
information upto current time due to memory constraints. Memory Y is
updated as a function of the past state of the memory and the current side
information, i.e., $N_i=f_{n,i}(N_{i-1},Y_i)$, where the $f_{n,i}$ is the memory
update function,
$f_{n,i}:\mathcal{N}_{i-1}\times\mathcal{X}\rightarrow\mathcal{N}_{i}$,
$i\in\mathbf{N}$. Here also the alphabet $\mathcal{N}_i$ can grow with $i$,
hence also includes the special case of complete memory, i.e.,
$N_i=Y^i$.
\item \textit{Source Decoder} : Source decoder uses the current encoded symbol, 
current  side information and the past memory states, to construct its estimate
of the source
symbol, i.e., $\hat{U}_i=f_{d,i}(X_i,Y_i,M_{i-1},N_{i-1})$, the decoding rule is
the map, $f_{d,i} :
\mathcal{X}\times\mathcal{Y}\times\mathcal{M}_{i-1}\times\mathcal{N}_{i-1}
\rightarrow\hat{\mathcal{U}}$.
The complete memory case corresponds to the decoding, $\hat{U}_i(X^i,Y^i)$.
\end{itemize}
The alphabets $\mathcal{U}, \mathcal{X}, \mathcal{A}_v, \mathcal{Y},
\mathcal{M}, \mathcal{N}$ are assumed to be finite. Note that the finiteness of
the alphabets implies
we may assume, without loss of generality, that $0\le\Lambda(\cdot)\le
\Lambda_{max}<\infty$ and
$0\le C(\cdot)\le \Gamma_{max}<\infty$. We make the further assumption that
there exists $a\in\mathcal{A}_v$ such
 that $C(a)=0$. 
Thus it makes sense to
consider cost constraints, $\Gamma\in[0,\Gamma_{max}]$.\par
Our approach to construction of the ACOE is similar
to that taken in previous sections, we first consider  the system with modified
source,
$\{V_i=U_i^{i+d}\}_{i\in\mathbf{N}}$ and it is equivalent to consider
source and action encoding rules as mappings,
$\{f_{e,i}(v,V^{i-1},Y^{i-1})\}_{v\in\mathcal{V}}$ and
$\{f_{v,i}(x,X^{i-1})\}_{x\in\mathcal{X}}$. Hence the modified vending machine
is,
\bea
P(Y_i|V_i,A_{v,i}(X_i))=P(Y_i|U_i,A_{v,i}(X_i)),
\eea
and the modified cost function is,
\bea
\tilde{\Lambda}(V_i,\hat{V}_i)=\Lambda(U_i,\hat{U}_i).
\eea
We study two scenarios under this setting :

\subsubsection{Complete Memory}
\label{sirtcactioncomplete}
Here $M_i=X^i$ and $N_i=Y^i$. Note that we can restrict our attention to optimal
decoders of the form, 
$\hat{U}_i(X^i,Y^i)=\hat{U}_{Bayes}(X^i,Y^i)$ (cf. Lemma \ref{lemma1}). Let us
denote the minimum expected average distortion achieved to be
$D^{FB}_a(d)$. Here $FB$ superscript indicating we have side information
available as a feedback to the encoder, $a$ subscript denotes presence of
actions
and $d$ stands for lookahead. We have the average cost optimality equation as,
\begin{algorithm}
\bea
&&\rho^{\lambda}(u_1,\cdots,u_{d+1},\beta)+h^{\lambda}(u_1,\cdots,u_{d+1},
\beta)\nonumber\\
&=&\max_{(a_e,a_v)}\left[g^{\lambda}(u_1,\cdots,u_{d+1},\beta,a_e,
a_v)+\sum_{ \tilde { u }\in\mathcal{U} , \tilde { x }\in\mathcal{X} ,
\tilde{y}\in\mathcal{Y}}P(\tilde{u})\1_{\{\tilde{x}=a_e(u_2,\cdots,u_{d+1},
\tilde { u } ) \}}
P(\tilde{y}|u_2,a_v(\tilde{x}))h(u_d,\cdots,u_{d+1},\tilde{u},\tilde{
\beta})\right
]\nonumber\\
&& \mbox{ }\forall\mbox{ }(u_1,\cdots,u_{d+1})\in\mathcal{U}^{d+1},
\beta\in\mathcal{P}(\mathcal{U}^{d+1}),\label{B4}
\eea
\end{algorithm}
\newpage
where $(a_e,a_v)\in\mathcal{A}_e\times\mathcal{A}_v$ and 
$\tilde{\beta}=G(\beta,a_e,a_v,(u_2^{d+1},\tilde{u},\tilde{x},\tilde{y}))$ (cf.
Appendix \ref{appendixA})
is the updated belief and $g^{\lambda}(\cdot)$ is the Lagrangian augmented cost,
\bea
&&g^{\lambda}(u_1,\cdots,u_{d+1},\beta,a_e,a_v)\nonumber\\
&=&g(u_1,\cdots,u_{d+1},\beta,
a_e,a_v)+\lambda\left(\Gamma-l(u_1,\cdots,u_{d+1},\beta,
a_e,a_v)\right)\\
&=&-\min_{\hat{u}\in\mathcal{U}}\sum_{\overline{u}\in\mathcal{U}}
\beta_1(\overline{u})\Lambda(\overline{u},\hat{u}) +\lambda\left(\Gamma-\sum_ {
\tilde
{ u }
\in\mathcal
{U},\tilde{x}\in\mathcal{X}}
P(\tilde{u})\1_{\{\tilde{x}=a_e(u_2,\cdots,u_{d+1},
\tilde { u} ) \}}C(A_v(\tilde{x})) \right).\label{g4}
\eea
We have now the
following theorem, with proof in Appendix \ref{appendixA}.
\begin{theorem}\label{theorem5}
For a fixed lookahead, $d$, let $(\rho^{\lambda}(\cdot),h^{\lambda}(\cdot))$
solves the ACOE Eq. (\ref{B4}). Then
the optimal average distortion is given by,
\bea
D^{FB}_a(d)=-\inf_{\lambda\ge
0}\sup_{x\in\mathcal{U}^{d+1}\times\mathcal{P}(\mathcal{U}^{d+1})}\rho^{\lambda}
(x).
\eea
\end{theorem}
\begin{note}
\label{notel}
Note that for a
fixed finite lookahead, $d\ge 1$, $D^{FB}_a(d)$ contrasts the minimum
distortion at, $d=0$, where symbol by symbol encoding, action-encoding and
decoding are optimal, i.e,
\bea\label{asymbol}
D^{FB}_a(0)=\min_{X:\mathcal{U}\rightarrow \mathcal{X},
A:\mathcal{X}\rightarrow\mathcal{A}}\E\left[\Lambda\left(U,\hat{U}_{Bayes}(P_{
U|X,Y})\right)\right ],
\eea
while at infinite lookahead, $d=\infty$, the minimum distortion is given by
the distortion rate function at unit rate by results from
\cite{HaimTsachyVendor}, i.e,
\bea
D^{FB}_a(\infty)&=&\min_{P_{A_v,W|U},\hat{U}^{opt}:\mathcal{W}\times\mathcal{Y}
\rightarrow\hat{\mathcal{U}} } \E\left [ \Lambda\left(U , \hat { U } ^ { opt }
(W,Y)\right)\right ]\nonumber\\
\mbox{such that}&&I(U;W,A_v)\le \log_{2}\card{\mathcal{X}}\nonumber\\
&&\card{\mathcal{W}}\le\card{\mathcal{U}}\card{\mathcal{A}_v}+2\nonumber\\
&&\E\left[C(A_v)\right]\le\Gamma,\label{ashannon}
\eea
where $I(\cdot;\cdot)$ is the mutual information (cf. \cite{CovThom}). The
above distortion is basically the distortion rate function (cf. Theorem 3
\cite{HaimTsachyVendor}) evaluated at the rate equal to the cardinality of the
alphabet $\card{\mathcal{X}}$. The proofs for Equations (\ref{asymbol}) and
(\ref{ashannon}) are similar to those of Lemma \ref{lemma3} and
Lemma \ref{lemma4}.
\end{note}

\subsubsection{Finite Memory}
\label{sirtcactionfinite}
In this section, all memories are finite (not growing with time). With the
object of minimizing the
expected distortion, we cast this problem as a constrained Markov decision
process. To be able to do that, for
reasons discussed in Section \ref{rtcfinite}, we assume, $f_{m,i}=f_m$,
$f_{n,i}=f_n$,
$\mathcal{M}_i=\mathcal{M}$ and $\mathcal{N}_i=\mathcal{N}$ for all
$i\in\mathbf{N}$, the alphabets $\mathcal{M}$, $\mathcal{N}$ being
finite. We further assume stationary optimal decoding and actuator policies,
i.e., $f_{d,i}(\cdot,\cdot,\cdot,\cdot)=U^{opt}(\cdot,\cdot,\cdot,\cdot)$, and
$f_{a,i}(\cdot,\cdot)=A_v^{opt}(\cdot,\cdot)$ for all $i\in\mathbf{N}$. \\
Fix a lookahead $d$. Now for fixed $\lambda\ge 0$, the average cost optimality
equation is,
\begin{algorithm}
\bea
&&\rho^{\lambda}(u_1,\cdots,u_{d+1},m,n)+h^{\lambda}(u_1,\cdots,u_{d+1},m,
n)\nonumber\\
&=&\max_{a\in\mathcal{A}}\left[g^{\lambda}(u_1,\cdots,u_{d+1},m,n,a)+\sum_{
\tilde { u }\in\mathcal{U} , \tilde { x }\in\mathcal{X} ,
\tilde{y}\in\mathcal{Y}}P(\tilde{u})\1_{\{\tilde{x}=a(u_2,\cdots,u_{d+1},
\tilde { u } ) \}}
P(\tilde{y}|u_2,A^{opt}_v(\tilde{x}))h(u_d,\cdots,u_{d+1},\tilde{u},
\tilde{m},\tilde{n}))\right
]\nonumber\\
&& \mbox{ }\forall\mbox{ }(u_1,\cdots,u_{d+1})\in\mathcal{U}^{d+1},
m\in \mathcal{M},n\in\mathcal{N},\label{B5}
\eea
\end{algorithm}
\\
where $\tilde{m}=f_m(m,\tilde{x})$ and $\tilde{n}=f_n(n,\tilde{y})$ are memory
updates and $g^{\lambda}(\cdot)$ is the Lagrangian augmented cost,
\bea
&&g^{\lambda}(u_1,\cdots,u_{d+1},m,n,a)\nonumber\\
&=&g(u_1,\cdots,u_{d+1},m,n,
a)+\lambda\left(\Gamma-l(u_1,\cdots,u_{d+1},m,n,
a)\right)\\
&=&-\sum_{\tilde{u}\in\mathcal{U},\tilde{x}\in\mathcal{X},\tilde{y}\in\mathcal{Y
}}P(\tilde{u})\1_{\{\tilde{x}=a(u_2,\cdots,u_{d+1},
\tilde { u } ) \}}
P(\tilde{y}|u_2,A^{opt}_v(\tilde{x}))\Lambda(u_2,\hat{U}_{Bayes}(\tilde{x},
\tilde{y},m,n))\nonumber\\
&+&\lambda\left(\Gamma-\sum_{\tilde{u}\in\mathcal{U},\tilde{x}\in\mathcal{X}}
P(\tilde{u})\1_{\{\tilde{x}=a_e(u_2,\cdots,u_{d+1},\tilde{u})\}}C(A_v^{opt}
(\tilde{x})) \right).\label{g5}
\eea
Let us denote the optimal distortion by $D^{FB}_a(d,\mathcal{M},\mathcal{N})$.
We
have now the following theorem, with proof in Appendix \ref{appendixB}.
\begin{theorem}\label{theorem6}
For a fixed lookahead, $d$, let $(\rho^{\lambda}(\cdot),h^{\lambda}(\cdot))$
solves the ACOE Eq. (\ref{B5}). Then
the optimal average distortion is given by,
\bea
D^{FB}_a(d,\mathcal{M},\mathcal{N})=-\inf_{\lambda\ge
0}\sup_{x\in\mathcal{U}^{d+1}\times\mathcal{M}\times\mathcal{N}}\rho^{\lambda}
(x).
\eea
\end{theorem}

\subsection{Encoder does not have access to Side Information}
\label{nosirtcaction}
Here encoder does not recieve any knowledge about side information. In this
section also, we make assumptions A1 and A2 and further assume finite state
decoders (for reasons similar to those outlined in Note \ref{note5}). For a
fixed lookahead $d$, $\lambda\ge 0$,
we have the average cost optimality equation,
\begin{algorithm}
\bea
&&\rho^{\lambda}(u_1,\cdots,u_{d+1},\beta,\gamma)+h^{\lambda}(u_1,\cdots,u_{d+1}
,\beta, \gamma)\nonumber\\
&=&\max_{a\in\mathcal{A}}\left[g^{\lambda}(u_1,\cdots,u_{d+1},\beta,\gamma,
a)+\sum_ {
\tilde { u }\in\mathcal{U} , \tilde { x }\in\mathcal{X} ,
\tilde{y}\in\mathcal{Y}}P(\tilde{u})\1_{\{\tilde{x}=a(u_2,\cdots,u_{d+1},
\tilde { u } ) \}}
P(\tilde{y}|u_2,A^{opt}_v(\tilde{x}))h(u_d,\cdots,u_{d+1},\tilde{u},
\tilde{\beta},\tilde{\gamma}))\right
]\nonumber\\
&& \mbox{ }\forall\mbox{ }(u_1,\cdots,u_{d+1})\in\mathcal{U}^{d+1},
\beta\in \mathcal{P}(\mathcal{M}),\gamma\in\mathcal{P}(\mathcal{N}),\label{B6}
\eea
\end{algorithm}
\\
where $\tilde{\beta}=\xi_m(\beta,u_1^{d+1},\tilde{u},a)$ and
$\tilde{\gamma}=\xi_n(\gamma,u_1^{d+1},\tilde{u},a)$ are
belief updates (cf. Appendix \ref{appendixC}) and $g^{\lambda}(\cdot)$ is the
Lagrangian augmented cost,
\bea
&&g^{\lambda}(u_1,\cdots,u_{d+1},\beta,\gamma,a)\nonumber\\
&=&g(u_1,\cdots,u_{d+1},\beta,\gamma,
a)+\lambda\left(\Gamma-l(u_1,\cdots,u_{d+1},\beta,\gamma,
a)\right)\\
&=&-\sum_{\tilde{m}\in\mathcal{M},\tilde{n}\in\mathcal{N},\tilde{u}\in\mathcal{U
} , \tilde { x } \in\mathcal { X } , \tilde { y } \in\mathcal { Y
}}\beta(\tilde{m})\gamma(\tilde{n})P(\tilde{u})\1_{\{\tilde{x}=a(u_2,\cdots,u_{
d+1 } ,
\tilde { u } ) \}}
P(\tilde{y}|u_2,A^{opt}_v(\tilde{x}))\Lambda(u_2,\hat{U}_{Bayes}(\tilde{x},
\tilde{y},\tilde{m},\tilde{n}))\nonumber\\
&+&\lambda\left(\Gamma-\sum_{\tilde{u}\in\mathcal{U},\tilde{x}
\in\mathcal { X }
}P(\tilde{u})\1_{\{\tilde{x}=a_e(u_2,\cdots,u_{d+1},\tilde{u})\}} C(A_v^ { opt }
(\tilde{x})) \right).\label{g6}
\eea
Let us denote the optimal distortion by $D^{NF}_a(d,\mathcal{M},\mathcal{N})$
(NF standing for no feedback of side information symbols).
We can now state the following theorem whose proof is defered to Appendix
\ref{appendixC}.
\begin{theorem}\label{theorem7}
For a fixed lookahead, $d$, suppose that 
$(\rho^{\lambda}(\cdot),h^{\lambda}(\cdot))$
solves the ACOE Eq. (\ref{B6}). Then
the optimal average distortion is given by,
\bea
D^{NF}_a(d,\mathcal{M},\mathcal{N})=-\inf_{\lambda\ge
0}\sup_{x\in\mathcal{U}^{d+1}\times\mathcal{P}(\mathcal{M})\times\mathcal{P}
(\mathcal{N})} \rho^{\lambda}(x).
\eea
\end{theorem}
\begin{note}
 Note that the ACOE in this section on sequential
source coding with lookahead and a side information vending machine is
amenable
to
computational solutions as in Section
\ref{rtcfinitecomputation}. Here also
$D^{FB}_a(d,\mathcal{M},\mathcal{N})$ can be computed for
increasing
memories exactly, and yield non trivial bounds on $D^{FB}_a(d)$.
\end{note}

\section{Summary of the Results}
\label{summary}
In this section we provide a summary of the
various settings considered in this paper on real time communication with fixed
finite lookahead at the
encoder, and the transformations performed to cast the problem as (constrained
or unconstrained) Markov decision process. The methodology is to
construct an average
cost optimality equation (ACOE), and seek its solution. We have considered two
classes of problems in this paper : 
\begin{enumerate}
 \item\textit{Real Time Communication, Fig.
\ref{feedbackmemoryrtclookahead}}. The problem is characterized by tuple
 $(\mathcal{S}, \mathcal{A}, \mathcal{W},F,P_{S},P_{W},g)$, the meaning of
various symbols being explained in Section \ref{controlledmarkov}. The
general ACOE is,
\bea
\lambda(s) + h(s) =
\max_{a\in\mathcal{A}}\left[g(s,a)+\sum_{w\in\mathcal{W}}P_W(w|s,a)h(F(s,a,
w))\right ] \mbox {
}\forall\mbox{ }s\in \mathcal{S}.
\eea
Note in all the settings we considered, $\sup$ is replaced by $\max$ as the
set of actions is finite. If $\exists\
\lambda^{\ast}\in\mathbf{R}$ and a
bounded $h(\cdot)$, satisfying the above equation, then using Theorem
\ref{theorem1}, the minimum distortion is $-\lambda^{\ast}$. The following table
exhibits the transformations, along with pointers to the equations in the paper
that cast the problem of Fig. \ref{feedbackmemoryrtclookahead} as an 
unconstrained
Markov decision process : 
\newline
\begin{center}
  \begin{tabular}{ | p{3.2cm} | p{3.5cm} | p{3.5cm} | p{3.5cm} | }
    \hline
    Real-Time Communication, Fig. \ref{feedbackmemoryrtclookahead}, Lookahead,
$d$ & Noise-Free Feedback, Complete Memory Decoding & Noise-free
Feedback, Finite Memory ($\mathcal{M}$) Decoder &  No Feedback, Finite
Memory ($\mathcal{M}$) Decoder \\ \hline    \hline
$\mathcal{S}$, state space & $\mathcal{U}^{d+1}\times
\mathcal{P}(\mathcal{U}^{d+1})$ & $\mathcal{U}^{d+1}\times \mathcal{M} $ &
$\mathcal{U}^{d+1}\times \mathcal{P}(\mathcal{M})$ \\
\hline
$\mathcal{A}$, action space & $\mbox{Mappings}:\mathcal{U}^{d+1}\rightarrow
\mathcal{X}$ & $\mbox{Mappings}:\mathcal{U}^{d+1}\rightarrow
\mathcal{X}$ & $\mbox{Mappings}:\mathcal{U}^{d+1}\rightarrow
\mathcal{X}$\\ \hline
$\mathcal{W}$, disturbance & $\mathcal{U}^{d+1}\times\mathcal{Y}$
&$\mathcal{U}^{d+1}\times\mathcal{Y}$  & $\mathcal{U}^{d+1}$\\
\hline
$F(\cdot)$ & Eq. (\ref{F1}) & Eq. (\ref{F2}) & Eq. (\ref{F3}) \\ \hline 
$P_W(\cdot|S,A)$ & Eq. (\ref{W1}) & Eq. (\ref{W2}) & Eq. (\ref{W3}) \\ \hline 
$g(S,A)$, reward &
Eq. (\ref{g1})&
Eq. (\ref{g2}) & Eq. (\ref{g3})\\ \hline
ACOE &
Eq. (\ref{B1})&
Eq. (\ref{B2}) & Eq. (\ref{B3})\\ \hline
  \end{tabular}
\newline
\end{center}

\item \textit{Source Coding with a Side Information Vending Machine, Fig.
\ref{actionfeedbackmemoryrtclookahead}} : The problem is characterized by tuple
 $(\mathcal{S},
\mathcal{A},\mathcal{W},F,P_{s},P_{w},g,l,\Gamma)$ explained
in Section \ref{unconstrainedcontrol}. Here also the
general ACOE is,
\bea
\rho^{\lambda}(s) + h^{\lambda}(s) =
\sup_{a\in\mathcal{A}}\left[g^{\lambda}(s,a)+\sum_{w\in\mathcal{W}}P(w|s,a)h(F(s
, a,
w))\right ] \mbox {
}\forall\mbox{ }s\in \mathcal{S},
\eea
where $\lambda$ is the Lagrangian parameter. Minimum distortion is given by the
Theorems \ref{theorem5}, \ref{theorem6} and \ref{theorem7}, respectively, for
the
cases tabulated below.\\
\begin{center}
  \begin{tabular}{ | p{3.2cm} | p{3.5cm} | p{3.5cm} | p{3.5cm} | }
    \hline
    Source Coding With SI \textquotedblleft Vendor\textquotedblright, Fig.
\ref{actionfeedbackmemoryrtclookahead}, Lookahead,
$d$ & Noise-Free Feedback, Complete Memory Decoding & Noise-free
Feedback, Finite Memory ($\mathcal{M},\mathcal{N}$) Decoder &  No Feedback,
Finite
Memory ($\mathcal{M},\mathcal{N}$) Decoder \\ \hline    \hline
$\mathcal{S}$, state space & $\mathcal{U}^{d+1}\times
\mathcal{P}(\mathcal{U}^{d+1})$ & $\mathcal{U}^{d+1}\times
\mathcal{M}\times \mathcal{N} $ &
$\mathcal{U}^{d+1}\times \mathcal{P}(\mathcal{M})\times
\mathcal{P}(\mathcal{N})$ \\
\hline
$\mathcal{A}$, action space & $\mbox{Mappings}:\mathcal{U}^{d+1}\times
\mathcal{X}\rightarrow
\mathcal{X}\times\mathcal{A}_v$&
$\mbox{Mappings}:\mathcal{U}^{d+1}\rightarrow
\mathcal{X}$  & $\mbox{Mappings}:\mathcal{U}^{d+1}\rightarrow
\mathcal{X}$\\ \hline
$\mathcal{W}$, disturbance &
$\mathcal{U}^{d+1}\times\mathcal{X}\times\mathcal{Y}$
&$\mathcal{U}^{d+1}\times\mathcal{X}\times\mathcal{Y}$  &
$\mathcal{U}^{d+1}$\\
\hline
$F(\cdot)$ & Eq. (\ref{F4}) (Appendix \ref{appendixA}) & Eq. (\ref{F5})
(Appendix \ref{appendixB})& Eq. (\ref{F6}) (Appendix \ref{appendixC})\\ \hline 
$P_w(\cdot|S,A)$ & Eq. (\ref{W4}) (Appendix \ref{appendixA})& Eq. (\ref{W5})
(Appendix \ref{appendixB})& Eq. (\ref{W6}) (Appendix \ref{appendixC})\\ \hline 
$g^{\lambda}(S,A)$, Lagrangian augmented reward & Eq. (\ref{g4}) (Appendix
\ref{appendixA}) &
Eq. (\ref{g5}) (Appendix \ref{appendixB})& Eq. (\ref{g6}) (Appendix
\ref{appendixC})\\ \hline
ACOE &
Eq. (\ref{B4})&
Eq. (\ref{B5}) & Eq. (\ref{B6})\\ \hline
  \end{tabular}
\end{center}
\end{enumerate}

\section{Conclusion}
\label{conclusion}
In this paper, we consider an important class of problems in real time
coding :  a memoryless source is to be communicated over a memoryless channel,
with sequential encoding and decoding and with a fixed finite lookahead of
future symbols available at the encoder. Unit delay feedback may or may
not be present, and decoding is based on the channel output symbols without
delay, with or without a memory constraint. In all these scenarios, under the
objective of minimizing
the per-symbol distortion, we obtain average cost
optimality equations whose solution yields the minimum achievable distortion, as
well as sufficient
conditions for the optimality of
stationary policies. We contrast the minimum distortion at a fixed lookahead,
with the best achievable with zero lookahead, where symbol by symbol
encoding-decoding is optimal,
and with the infinite lookahead case, for which the minimum achievable per
symbol distortion is shown to coincide with that for the classical joint source
channel coding problem, where separation is optimal. For the Bernoulli
source and binary symmetric channel under Hamming loss, in case of finite state
decoders, we compute
exactly the minimum distortion values for various memory sizes, and study the 
upper bounds that they yield on the minimum distortion for a fixed lookahead in
the absence of memory constraints. Answering
the question
\textquotedblleft to look or not to lookahead\textquotedblright, we 
characterize general conditions on the source and channel such that symbol by
symbol encoding-decoding is optimal  within the class of schemes of a given
lookahead. We obtain and plot the region for source and
channel parameters in case of Bernoulli source, binary symmetric channel
and Hamming distortion, where the symbol by symbol policy is strictly
suboptimal. We then demonstrate that this framework of casting real time
coding problems as Markov decision problems with average
cost criteria can be useful in various other frameworks by applying this same
methodology in source coding problem with a side information vending
machine, where encoder encodes the source sequentially, with a possible
lookahead, decoder takes cost constrained actions to receive the side
information about the source. This setting is cast as a
constrained Markov decision problem and it is shown that a stationary
randomized policy can attain the minimum per-symbol distortion which is
characterized as the
solution to a saddle point equation. 
\section*{Acknowledgment}
The authors would like to thanks Benjamin Van Roy for enlightening
discussions. This work is supported by The Scott A. and Geraldine D. Macomber
Stanford Graduate Fellowship and NSF Grants CCF-1049413 and 4101-38047. The
authors also acknowledge the support of Center for Science of Information
(CSoI), an NSF Science and Technology Center, under grant agreement CCF-0939370.
\bibliography{rtc}
\bibliographystyle{IEEEtran}

\appendices

\section{} 
\label{appendix0}
In Section \ref{problem} the minimum expected distortion is defined as, 
\bea
D(d)=\inf_{\{f_e,f_m,f_d\}}\limsup_{N\rightarrow\infty}\E\left[\frac{1}{N}
\sum_{i=1}^{N}\Lambda(U_i,\hat{U}_i)\right].
\eea
Note that $\inf$ in the above definition can be replaced by $\min$ over the
class of $(f_e,f_m,f_d)$-policies as outlined in Section \ref{problem}. This is
argued by constructing an $(f_e,f_m,f_d)$-policy that achieves $D(d)$. Fix
lookahead
$d$. As $D(d)$ always exists (also it is finite due to our assumption that
$\Lambda(\cdot,\cdot)\le\Lambda_{max}<\infty$), for a positive non-increasing
vanishing sequence $\{\epsilon_m\}_{m\ge 1}$, we can construct a sequence of
policies, i.e., $\{\mu_m\}_{m\ge 1}$ with expected average
distortion $D_{\mu_m}(d)$, i.e.,
\bea
D_{\mu_m}(d)&=&\limsup_{N\rightarrow\infty}\E_{\mu_m}\left [ \frac { 1 }
{ N }
\sum_{i=1}^{N}\Lambda(U_i,\hat{U}_i)\right]\\
&=&\limsup_{N\rightarrow\infty}D^N_{
\mu_m}(d),
\eea
($\E_{\mu_m}$ is the expectation with respect to the joint probability
distribution induced
when the policy used is $\mu_m$), such that $D_{\mu_m}(d)$ is a monotone
non-increasing sequence converging to $D(d)$. By the definition of $\limsup$,
for
every $m\ge 1$, $\exists\ N_m(\epsilon_m)$ such that $\forall\ N\ge
N_m$ ($N_m$ being function of $\epsilon_m$ is implied henceforth),
\bea
D^{N}_{\mu_m}(d)\le D_{\mu_m}(d)+\epsilon_m.
\eea 
Now define a block-length sequence $\{l_m\in\mathbf{N}\}_{m\ge 1}$, satisfying
the following requirements,
\begin{itemize}
 \item \textit{R1} : $l_m>\sum_{i=1}^{m-1}l_i$ and that 
$\frac{\sum_{i=1}^{m-1}l_i}{l_m}\rightarrow 0$ as $m\rightarrow\infty$.
\item \textit{R2} : $l_m>\max\{N_m,N_{m+1}\}$ and that
$\frac{\max\{N_m,N_{m+1}\}}{l_m}\rightarrow 0$ as $m\rightarrow\infty$.
\end{itemize}
Note that we can always choose such a sequence, for eg.
$l_m=\max\{N_m,N_{m+1}\}(\sum_{i=1}^{m-1}l_i)$. We define a block-coding scheme
$\mu^{\ast}$ which operates with block length $l_i$ in
$i^{th}$ block with scheme $\mu_i$. Operating this scheme for time
$N\in(\sum_{i=1}^{m}l_i,\sum_{i=1}^{m+1}l_i]$ for some $m=m(N)$, (note
$m(N)\rightarrow\infty$ as $N\rightarrow\infty$) we can bound the normalized
distortion as:
\begin{itemize}
 \item \textit{(Case 1)} $N-\sum_{i=1}^{m}l_i< N_{m+1}$

\bea
D^N_{\mu^{\ast}}(d)&\stackrel{(a)}{\le}&
\left(\frac{\sum_{i=1}^{m-1}l_i}{N}\right)\Lambda_{max}+\frac{l_m}{N}
\left(D_{\mu_{m}}
(d)+\epsilon_ {
m}\right)+\left(\frac{N-\sum_{i=1}^{m}l_i}{N}\right)\Lambda_{max}\\
&=&\left(\frac{N-l_m}{N}\right)\Lambda_{max}+\frac{l_m}{N}
\left(D_{\mu_{m}}
(d)+\epsilon_ {
m}\right)\\
&\le&\left(\frac{\sum_{i=1}^{m-1}l_i+N_{m+1}}{l_m}
\right)\Lambda_ { max } +D_{\mu_{m}}
(d)+\epsilon_ {
m},
\eea
where (a) is due to the fact that
$l_m>N_m$ (requirement
\textit{R2}) and hence distortion in block $m$ is bounded above by
$D_{\mu_m}(d)+\epsilon_m$. 

\item \textit{(Case 2)} $N-\sum_{i=1}^{m}l_i\ge N_{m+1}$
\bea
D^{N}_{\mu^{\ast}}(d)&\stackrel{(b)}{\le}&
\left(\frac{\sum_{i=1}^{m-1}l_i}{N}\right)\Lambda_{max}+\frac{l_m}{N}
\left(D_{\mu_{m}}
(d)+\epsilon_ {
m}\right)+\left(\frac{N-\sum_{i=1}^{m}l_i}{N}\right)\left(D_{\mu_{m+1}}
(d)+\epsilon_ { m+ 1}\right)\\
&\stackrel{(c)}{\le}&\left(\frac{\sum_{i=1}^{m-1}l_i}{N}\right)\Lambda_{max}
+\frac { l_m } { N }
\left(D_{\mu_{m}}
(d)+\epsilon_ {
m}\right)+\left(\frac{N-l_m}{N}\right)\left(D_{\mu_{m}}
(d)+\epsilon_ { m}\right)\\
&=&\left(\frac{\sum_{i=1}^{m-1}l_i}{N}\right)\Lambda_{max}+D_{\mu_{m}}
(d)+\epsilon_m\\
&\le&\left(\frac{\sum_{i=1}^{m-1}l_i+N_{m+1}}{l_m}
\right)\Lambda_ { max } +D_{\mu_{m}}
(d)+\epsilon_ {
m},
\eea
where (b) follows from bounding the distortion in block $m$ as in (a) and
similarly as $N-\sum_{i=1}^{m}l_i\ge N_{m+1}$, bounding distortion in block
$m+1$ by $D_{\mu_{m+1}}(d)+\epsilon_{m+1}$ and (c) follows from the fact that
both $D_{\mu_m}(d)$ and $\epsilon_m$ are non-increasing sequences. 
\end{itemize}
Thus we see in both the above cases, for any time $N$, the normalized
distortion is bounded above as,
\bea
D^N_{\mu^{\ast}}(d)&\le&\left(\frac{\sum_{i=1}^{m-1}l_i+N_{m+1}}{l_m}
\right)\Lambda_ { max } +D_{\mu_{m}}
(d)+\epsilon_ {
m},
\eea
which implies that the expected average distortion under policy
$\mu^{\ast}$ is,
\bea
D_{\mu^{\ast}}(d)&\le&
\limsup_{N\rightarrow\infty}\left[\left(\frac{\sum_{i=1}^{m-1}l_i+N_{m+1}}{l_m}
\right)\Lambda_{max}+D_{\mu_{m}}
(d)+\epsilon_m\right]\\
&\le&
\limsup_{N\rightarrow\infty}\left(\frac{\sum_{i=1}^{m-1}l_i}{l_m}
\right)\Lambda_{max}+\limsup_{N\rightarrow\infty}\left(\frac{N_{m+1}}{l_m}
\right)\Lambda_{max}+\limsup_{ N\rightarrow\infty } \left(D_{\mu_{m}}
(d)+\epsilon_m\right)\\
&\stackrel{(d)}{=}&D(d),
\eea
where (d) follows from the fact that 
$\frac{\sum_{i=1}^{m-1}l_i}{l_m}\rightarrow 0$ and
$\frac{N_{m+1}}{l_m}\rightarrow 0$ by requirements
\textit{R1} and \textit{R2} respectively, since $m(N)\rightarrow\infty$ as
$N\rightarrow\infty$.
Thus we have a scheme $\mu^{\ast}$ with minimum expected distortion
$D_{\mu^{\ast}}(d)\le D(d)$, but we know for any scheme $\mu^{\ast}$, $D(d)\le
D_{\mu^{\ast}}(d)$, implying $D_{\mu^{\ast}}(d)=D(d)$.

\section{Proof of Theorem \ref{theorem5}}\label{appendixA}
We will first obtain the ACOE Eq. (\ref{B4}). Define
the state sequence, $S_i=(V_i,\{P_{V_i=v|X^i,Y^i}\}_{v\in\mathcal{V}})$,
disturbance sequence, $W_i=(V_i,X_i,Y_i)$, action sequence,
$A_i=\{f_{e,i}(v,V^{i-1},Y^{i-1}),f_{v,i}(x,X^{i-1})\}_{v\in\mathcal{V},
x\in\mathcal{X}}$ is clearly a history dependent action, i.e. function of
$W^{i-1}=(V^{i-1},X^{i-1},Y^{i-1})$. We will now verify the conditions for the
defined state sequence, disturbance and action sequence to form a controlled
markov process. With some abuse of notation, we denote,
\bea
\beta_i&=&\{P_{V_i=v|X^i,Y^i}\}_{v\in\mathcal{V}}\\
A_{e,i}&=&f_{e,i}(\cdot,V^{i-1},Y^{i-1})\\
A_{v,i}&=&f_{v,i}(\cdot,X^{i-1}).
\eea
Now,
\bea
P(W_i|W^{i-1},S^{i-1},A^i)&=&P(V_i,X_i,Y_i|V^{i-1},X^{i-1},Y^{i-1},A_v^i,A_e^i,
\beta^i)\\
&=&K(V_{i-1},V_i)\1_{\{X_i=A_{e,i}(V_i)\}}P(Y_i|V_i,A_{v,i}(X_i))\\
&=&P_W(W_i|S_{i-1},A_i)\label{W4}
\eea
\bea
\beta_i&=&\left\lbrace\frac{\sum_{V_{i-1}}\beta_{i-1}(V_{i-1})
K(V_{i-1},v)\1_{\{X_i=A_{e,i}(v)\}}P(Y_i|v,A_{v,i}(X_i))}{\sum_{V_{i-1},V_i}
\beta_{i-1}(V_{i-1})
K(V_{i-1},V_i)\1_{\{X_i=A_{e,i}(V_i)\}}P(Y_i|V_i,A_{v,i}(X_i))}\right\rbrace_{
v\in\mathcal{V}} \\
&=&G(\beta_{i-1},A_{e,i},A_{v,i},V_i,X_i,Y_i)\\
&=&G(\beta_{i-1},A_i,W_i),
\eea
which implies,
\bea\label{F4}
S_i=F(S_{i-1},A_i,W_i).
\eea
The optimal decoding is $\hat{V}_{Bayes}(P_{V_i|X^i,Y^i})$. Also let
$g(S_i,A_{i+1})=-\E\left[\tilde{\Lambda}(V_i,\hat{V}_{Bayes}(P_{V_i|X^i,Y^i}
))\Big{|}Y^i\right]$ so
that we  have
\bea
\inf\limsup_{N\rightarrow\infty}\E\left[\frac{1}{N}\sum_{i=1}^{N}\tilde{
\Lambda }(V_{i},
\hat{V}_{Bayes}(X^i,Y^i)\right
]&=&-\sup\liminf_{N\rightarrow\infty}\E\left[\frac{1}{N}\sum_{i=1}^{N}g(S_{i-1},
A_i)\right ].
\eea
Also for the cost constraint on action,
\bea
\E\left[{C}(A_{v,i}(X_i))\Big{|}V^{i-1},X^
{i-1},Y^{i-1},A^i\right]&=&\sum_{\tilde{v}\in\mathcal{V},\tilde{x}\in\mathcal{X}
} K(V_{i-1},\tilde{v}
)\1_{\{\tilde{x}=A_{e,i}(\tilde{v})\}}{C}(A_{v,i}(\tilde{x}))\\
&=&l(S_{i-1},A_i),
\eea
which imply,
\bea
\limsup_{N\rightarrow\infty}\E\left[\frac{1}{N}\sum_{i=1}^{N}C(A_{v,i})\right ]
=\limsup_{N\rightarrow\infty}\E\left[\frac{1}{N}\sum_{i=1}^{N}l(S_{i-1},
A_i)\right].
\eea
Thus the problem of minimizing the average distortion subject to constraints on
the vending action is equivalent to a constrained Markov decision process,
$(\mathcal{S},\mathcal{A},\mathcal{W},F,P_{S},P_{W},g,l,
\Gamma)$ (note here the number of constraints is $k=1$). Fix a lookahead $d$.
Let $\beta_1$ denote the
marginal of belief $\beta$ with respect to the first argument. Now for fixed
$\lambda\ge 0$,
we have the average cost optimality equation as,
\bea
&&\rho^{\lambda}(u_1,\cdots,u_{d+1},\beta)+h^{\lambda}(u_1,\cdots,u_{d+1},
\beta)\nonumber\\
&=&\max_{
(a_e,a_v)}\left[g^{\lambda}(u_1,\cdots,u_{d+1},\beta,a_e,
a_v)+\sum_{ \tilde { u }\in\mathcal{U} , \tilde { x }\in\mathcal{X} ,
\tilde{y}\in\mathcal{Y}}P(\tilde{u})\1_{\{\tilde{x}=a_e(u_2,\cdots,u_{d+1},
\tilde { u } ) \}}
P(\tilde{y}|u_2,a_v(\tilde{x}))h(u_d,\cdots,u_{d+1},\tilde{u},\tilde{
\beta})\right
]\nonumber\\
&& \mbox{ }\forall\mbox{ }(u_1,\cdots,u_{d+1})\in\mathcal{U}^{d+1},
\beta\in\mathcal{P}(\mathcal{U}^{d+1}),
\eea
where $\tilde{\beta}=G(\beta,a_e,a_v,(u_2^{d+1},\tilde{u},\tilde{x},\tilde{y}))$
is the updated belief, $(a_e,a_v)\in\mathcal
{A}_e\times\mathcal{A}_v$ and $g^{\lambda}(\cdot)$ is the Lagrangian
augmented cost,
\bea
&&g^{\lambda}(u_1,\cdots,u_{d+1},\beta,a_e,a_v)\nonumber\\
&=&g(u_1,\cdots,u_{d+1},\beta,
a_e,a_v)+\lambda\left(\Gamma-l(u_1,\cdots,u_{d+1},\beta,
a_e,a_v)\right)\\
&=&-\min_{\hat{u}\in\mathcal{U}}\sum_{\overline{u}\in\mathcal{U}}
\beta_1(\overline{u})\Lambda(\overline{u},\hat{u}) +\lambda\left(\Gamma-\sum_ {
\tilde
{ u }
\in\mathcal
{U},\tilde{x}\in\mathcal{X}}
P(\tilde{u})\1_{\{\tilde{x}=a_e(u_2,\cdots,u_{d+1},
\tilde { u} ) \}}C(a_v(\tilde{x})) \right).\label{g4}
\eea
Now having obtained the ACOE, the proof is an application of Theorem
\ref{theorem2} stated in Section
\ref{unconstrainedcontrol}. We need merely verify that the
conditions :  
\begin{itemize}
 \item[C1] holds as the state space and actions space both are compact subsets
of
Borel spaces.
 \item[C2] holds because of our definitions of $g(\cdot)$, $l(\cdot)$ and
assumptions on cost and distortion constraints.
 \item[C3] Denoting the state by
$s=(u_1,\cdots,u_{d+1},\beta)$ and
action $a=(a_e,a_v)$, we have the stochastic kernel,
\bea
Q(\tilde{s}|s,a)&=&\sum_{\tilde { x
}\in\mathcal{X} ,
\tilde{y}\in\mathcal{Y}}P(\tilde{u})\1_{\{\tilde{x}=a_e(u_2,\cdots,u_{d+1},
\tilde { u } ) \}}
P(\tilde{y}|u_2,a_v(\tilde{x}))\1_{\{\tilde{\beta}=G(\beta,a,(u_2^{d+1},\tilde{
u},\tilde{x},\tilde{y})\}}
\nonumber\\&&\mbox{ if
}\tilde{s}=(u_2,\cdots,u_{d+1},\tilde{u},\tilde{\beta}),\\
&=& 0 \mbox{ otherwise }.
\eea
Fix tuple $(u_1^{d+1},a)$ which takes values in a finite set. Consider a
sequence
$\beta_n\rightarrow\beta$. Let $\mu_n$ and $\mu$ be the measure on
$\mathcal{B}(\mathcal{S})$ induced by $Q(\cdot|u_1^{d+1},\beta_n,a)$ and
$Q(\cdot|u_1^{d+1},\beta,a)$ respectively. Proving C3
is equivalent to proving that $\forall\ h\in C_b(\mathcal{S})$, we have
$\mu_n(h)\rightarrow\mu(h)$, i.e.
\bea
&&\sum_{ \tilde { u }\in\mathcal{U} , \tilde { x
}\in\mathcal{X} ,
\tilde{y}\in\mathcal{Y}}P(\tilde{u})\1_{\{\tilde{x}=a_e(u_2,\cdots,u_{d+1},
\tilde { u } ) \}}
P(\tilde{y}|u_2,a_v(\tilde{x}))h(F(\beta_n,a,(u_2^{d+1},\tilde{
u},\tilde{x},\tilde{y}))\nonumber\\
&&\rightarrow \sum_{ \tilde { u }\in\mathcal{U} , \tilde { x
}\in\mathcal{X} ,
\tilde{y}\in\mathcal{Y}}P(\tilde{u})\1_{\{\tilde{x}=a_e(u_2,\cdots,u_{d+1},
\tilde { u } ) \}}
P(\tilde{y}|u_2,a_v(\tilde{x}))h(F(\beta,a,(u_2^{d+1},\tilde{
u},\tilde{x},\tilde{y})),
\eea
which is true as $F(\cdot)$ (by its definition Eq. (\ref{F4})) is continuous
in its arguments.
\item[C4] (Slater's Condition) We need to show there exists a policy such that
the constraint on the vending action are strictly satisfied, but this is
trivially true as we can select a policy with $A_{v,i}(\cdot)$ such that
$C(A_{v,i})=0$, $\forall\ i$ which satisfies the slater's condition. Thus
C1-C4 being true, this
implies
that the optimal distortion, $D^{FB}_a(d)$ is,
\bea
D^{FB}_a(d)=-\rho^{\ast}&\stackrel{(a)}{=}&-\sup_{(\nu,\pi)\in
P(\mathcal{S})\times\Pi_{HD}}\inf_{\lambda\ge 0}L((\nu,\pi),\lambda)\\
&\stackrel{(b)}{=}&-\inf_{\lambda\ge 0}\sup_{(\nu,\pi)\in
P(\mathcal{S})\times\Pi_{HD}}L((\nu,\pi),\lambda)\\
&\stackrel{(c)}{=}&-\inf_{\lambda\ge
0}\sup_{x\in\mathcal{U}^{d+1}\times\mathcal{M}(\mathcal{U}^{d+1})}\rho^{\lambda}
(x),
\eea
where (a) follows from the definition of $\rho^{\ast}$, (b) follows from
Theorem \ref{theorem2} (note assumptions C1-C4 are satisfied here as proved
above) while (c) follows as $(\rho^{\lambda}(\cdot),h^{\lambda}(\cdot))$ solve
the ACOE.
\end{itemize}

\section{Proof of Theorem \ref{theorem6}}\label{appendixB}
Define the state sequence as, $S_i=(V_i,M_i,N_i)$ and the disturbance sequence,
$W_i=(V_i,X_i,Y_i)$. We will first derive the ACOE Eq. (\ref{B5}). The action
(encoder's control) sequence is history
dependent,
$A_{i}=f_{e,i}(\cdot,W^{i-1})$. (Note here $A_i$ is the encoding action while
$A_v^{opt}$ is the action taken by the decoder to observe side information). 
It can  be easily established as in previous sections that,
\bea
P(W_i|W^{i-1},S^{i-1},A^i)=P_W(W_i|S_{i-1},A_{i})=K(V_{i-1},V_i)\1_{\{X_i=A_{i
} (V_i)\}}P(Y_i|V_i,A^{opt}_{v}(X_i))\label{W5}.
\eea
We have,
 \bea\label{F5}
 X_i=A_{i}(V_i)\mbox{, }M_i=f_m(M_{i-1},X_i)\mbox{,
}N_i=f_n(N_{i-1},Y_i)\mbox{, and hence }S_i=F(S_{i-1},A_{i},W_i).
\eea
By the assumptions in the Section \ref{sirtcactionfinite}, the decoding is
stationary, hence we have, 
\bea
&&\E\left[\tilde{\Lambda}(V_i,\hat{V}^{opt})|V^{i-1},M^{i-1},N^{i-1},Y^{i-1},
A^i\right ] \\
&=&\sum_{\tilde{v}\in\mathcal{V},\tilde{x}\in\mathcal{X},\tilde{y}\in\mathcal{Y}
} K(V_ { i-1 } , \tilde { v }
)\1_{\{\tilde{x}=A_i(\tilde{v})\}}P(\tilde{y}|\tilde{v},A_{v}^{opt}(\tilde{x}
))\Lambda(\tilde { v } , \hat { U } ^ { opt }
(\tilde{x},\tilde{y},M_{i-1},N_{i-1}))\\
&=&-g(S_{i-1},A_i).
\eea
For the cost constraints we have,
\bea
&&\E\left[C(A_v^{opt}(V_i)|V^{i-1},M^{i-1},N^{i-1},Y^{i-1},
A^i\right ] \\
&=&\sum_{\tilde{v}\in\mathcal{V},\tilde{x}\in\mathcal{X}} K(V_ { i-1 } , \tilde
{ v }
)\1_{\{\tilde{x}=A_i(\tilde{v})\}}C(A_v^{opt}(\tilde{x}))\\
&=&l(S_{i-1},A_i).
\eea
Fix a lookahead $d$. Now for fixed $\lambda\ge 0$, the average cost optimality
equation is,
\bea
&&\rho^{\lambda}(u_1,\cdots,u_{d+1},m,n)+h^{\lambda}(u_1,\cdots,u_{d+1},m,
n)\nonumber\\
&=&\max_{a\in\mathcal{A}}\left[g^{\lambda}(u_1,\cdots,u_{d+1},m,n,a)+\sum_{
\tilde { u }\in\mathcal{U} , \tilde { x }\in\mathcal{X} ,
\tilde{y}\in\mathcal{Y}}P(\tilde{u})\1_{\{\tilde{x}=a(u_2,\cdots,u_{d+1},
\tilde { u } ) \}}
P(\tilde{y}|u_2,A^{opt}_v(\tilde{x}))h(u_d,\cdots,u_{d+1},\tilde{u},
\tilde{m},\tilde{n}))\right
]\nonumber\\
&& \mbox{ }\forall\mbox{ }(u_1,\cdots,u_{d+1})\in\mathcal{U}^{d+1},
m\in \mathcal{M},n\in\mathcal{N},
\eea
where $\tilde{m}=f_m(m,\tilde{x})$ and $\tilde{n}=f_n(n,\tilde{y})$ are memory
updates and $g^{\lambda}(\cdot)$ is the Lagrangian augmented cost,
\bea
&&g^{\lambda}(u_1,\cdots,u_{d+1},m,n,a)\nonumber\\
&=&g(u_1,\cdots,u_{d+1},m,n,
a)+\lambda\left(\Gamma-l(u_1,\cdots,u_{d+1},m,n,
a)\right)\\
&=&-\sum_{\tilde{u}\in\mathcal{U},\tilde{x}\in\mathcal{X},\tilde{y}\in\mathcal{Y
}}P(\tilde{u})\1_{\{\tilde{x}=a(u_2,\cdots,u_{d+1},
\tilde { u } ) \}}
P(\tilde{y}|u_2,A^{opt}_v(\tilde{x}))\Lambda(u_2,\hat{U}_{Bayes}(\tilde{x},
\tilde{y},m,n))\nonumber\\
&+&\lambda\left(\Gamma-\sum_{\tilde{u}\in\mathcal{U},\tilde{x}\in\mathcal{X}}
P(\tilde{u})\1_{\{\tilde{x}=a_e(u_2,\cdots,u_{d+1},\tilde{u}\}}C(A_v^{opt}
(\tilde{x})) \right).\label{g5}
\eea
Once we have the ACOE, rest of the proof is similar to the proof of
Theorem
\ref{theorem5} by invoking Theorem \ref{theorem2}.

\section{Proof of Theorem \ref{theorem7}}\label{appendixC}
The proofs of this section follow in line with the
previous sections. We just need to establish the ACOE Eq. (\ref{B6}), rest
of the proof follows invoking Theorem \ref{theorem2}. Define the following : 
\bea
S_i&=&(V_i,P_{M_i|V^i},P_{N_i|V^i})\\
W_i&=&V_i\\
A_i&=&f_{e,i}(\cdot,W^{i-1}).
\eea
Let us use the following notation, $\beta_i=P(M_i|V^i)$ and
$\gamma_i=P(N_i|V^i)$. It is easy to see (along the lines of analysis in
previous sections),
\bea
P(W_i|W^{i-1},S^{i-1},A^i)=P(W_i|S_{i-1},A_{i})=K(V_{i-1},V_i),\label{W6}
\eea
and, 
\bea
\beta_i&=&\left\lbrace\frac{\sum_{M_{i-1},X_i}\beta_{i-1}(M_{i-1})
K(V_{i-1},V_i)\1_{\{X_i=A_i(V_i)\}}\1_{\{m=f_m(X_i,M_{i-1})\}}}{\sum_{M_{i-1},
X_i,M_i }
\beta_{i-1}(M_{i-1})
K(V_{i-1},V_i)\1_{\{X_i=A_i(V_i)\}}\1_{\{M_i=f_m(X_i,M_{i-1})\}}}\right\rbrace_{
m\in\mathcal{M}} \\
&=&\xi_m(\beta_{i-1},V_{i-1},V_i,A_i)\\
\gamma_i&=&\left\lbrace\frac{\sum_{N_{i-1},X_i,Y_i}\beta_{i-1}(N_{i-1})
K(V_{i-1},V_i)\1_{\{X_i=A_i(V_i)\}}P(Y_i|V_i,A^{opt}_v(X_i))
\1_{\{n=f_n(Y_i,N_{i-1})\}}}{\sum_{N_{i-1},X_i,Y_i,N_i}
\beta_{i-1}(N_{i-1})
K(V_{i-1},V_i)\1_{\{X_i=A_i(V_i)\}}P(Y_i|V_i,A^{opt}_v(X_i))
\1_{\{N_i=f_n(Y_i,N_{i-1})\}}}
\right\rbrace_{
n\in\mathcal{N}} \\
&=&\xi_n(\gamma_{i-1},V_{i-1},V_i,A_i),
\\&&\mbox{ which imply that }S_i=F(S_{i-1},A_i,W_i).\label{F6}
\eea
Also for constraints, 
\bea
&&\E\left[\tilde{\Lambda}(V_i,\hat{V}^{opt}(X_i,Y_i,M_{i-1},N_{i-1}))|V^{i-1},
\beta^
{i-1},\gamma^{i-1},Y^{i-1},A^i\right ]
\\
&=&\sum_{\tilde{m}\in\mathcal{M},\tilde{n}\in\mathcal{N},\tilde{v}\in\mathcal{V}
,
\tilde{x}\in\mathcal{X},\tilde { y }
\in\mathcal{Y}}\beta_{i-1}(\tilde{m})\gamma_{i-1}(\tilde{n})K(V_{i-1}
, \tilde { v } )\1_{\{\tilde{x}=A_i(\tilde{v})\}}P(\tilde { y }
|A^{opt}_v(\tilde {
x}),\tilde{v})\tilde{\Lambda}(\tilde{v} , \hat{V}^{opt}(\tilde{x},
\tilde{y},\tilde{m},\tilde{n})))\\
&=&-g(V_{i-1},\beta_{i-1},\gamma_{i-1},A_i)\\
&=&-g(S_{i-1},A_i),
\eea
and
\bea
\E\left[\tilde{C}(A_v^{opt}(V_i))|V^{i-1},\beta^
{i-1},\gamma^{i-1},Y^{i-1},A^i\right ]
&=&\sum_{\tilde{x}\in\mathcal{X},\tilde{v}\in\mathcal{V}
}K(V_{i-1},\tilde{v}
)\1_{\{\tilde{x}=A_i(\tilde{v})\}}\tilde{C}(A_v^{opt}(\tilde{x}))\\
&=&l(S_{i-1},A_i).
\eea
After these transformations for a fixed lookahead $d$, $\lambda\ge 0$,
we have the average cost optimality equation,
\bea
&&\rho^{\lambda}(u_1,\cdots,u_{d+1},\beta,\gamma)+h^{\lambda}(u_1,\cdots,u_{d+1}
,\beta, \gamma)\nonumber\\
&=&\max_{a\in\mathcal{A}}\left[g^{\lambda}(u_1,\cdots,u_{d+1},\beta,\gamma,
a)+\sum_ {
\tilde { u }\in\mathcal{U} , \tilde { x }\in\mathcal{X} ,
\tilde{y}\in\mathcal{Y}}P(\tilde{u})\1_{\{\tilde{x}=a(u_2,\cdots,u_{d+1},
\tilde { u } ) \}}
P(\tilde{y}|u_2,A^{opt}_v(\tilde{x}))h(u_d,\cdots,u_{d+1},\tilde{u},
\tilde{\beta},\tilde{\gamma})\right
]\nonumber\\
&& \mbox{ }\forall\mbox{ }(u_1,\cdots,u_{d+1})\in\mathcal{U}^{d+1},
\beta\in \mathcal{P}(\mathcal{M}),\gamma\in\mathcal{P}(\mathcal{N}), 
\eea
where $\tilde{\beta}=\xi_m(\beta,u_1^{d+1},\tilde{u},a)$ and
$\tilde{\gamma}=\xi_n(\gamma,u_1^{d+1},\tilde{u},a)$ are
belief updates and $g^{\lambda}(\cdot)$ is the Lagrangian augmented cost,
\bea
&&g^{\lambda}(u_1,\cdots,u_{d+1},\beta,\gamma,a)\nonumber\\
&=&g(u_1,\cdots,u_{d+1},\beta,\gamma,
a)+\lambda\left(\Gamma-l(u_1,\cdots,u_{d+1},\beta,\gamma,
a)\right)\\
&=&-\sum_{\tilde{m}\in\mathcal{M},\tilde{n}\in\mathcal{N},\tilde{u}\in\mathcal{U
} , \tilde { x } \in\mathcal { X } , \tilde { y } \in\mathcal { Y
}}\beta(\tilde{m})\gamma(\tilde{n})P(\tilde{u})\1_{\{\tilde{x}=a(u_2,\cdots,u_{
d+1 } ,
\tilde { u } ) \}}
P(\tilde{y}|u_2,A^{opt}_v(\tilde{x}))\Lambda(u_2,\hat{U}_{Bayes}(\tilde{x},
\tilde{y},\tilde{m},\tilde{n}))\nonumber\\
&+&\lambda\left(\Gamma-\sum_{\tilde{u}\in\mathcal{U},\tilde{x}
\in\mathcal { X }
}P(\tilde{u})\1_{\{\tilde{x}=a_e(u_2,\cdots,u_{d+1},\tilde{u})\}} C(A_v^ { opt }
(\tilde{x})) \right),\label{g6}
\eea
thus the ACOE Eq. (\ref{B6}) is established. 
\end{document}